\documentclass{article}
\usepackage{ijcai24}

\usepackage{times}
\usepackage{soul}
\usepackage{url}
\usepackage[hidelinks]{hyperref}
\usepackage[utf8]{inputenc}
\usepackage[small]{caption}
\usepackage{graphicx}
\usepackage{amsmath}
\usepackage{amsthm}
\usepackage{booktabs}
\usepackage{algorithm}
\usepackage{algorithmic}
\usepackage[switch]{lineno}

\usepackage{todonotes}
\usepackage{tabularx}
\usepackage{pgf, tikz}
\usetikzlibrary{arrows, automata,fit}
\usepackage{booktabs}
\usepackage{mathtools}
\usepackage{amssymb}
\usepackage{multirow}
\usepackage{hyperref}
\usepackage{comment}

\newcommand{\longTitle}{Toward Completing the Picture of Control in Schulze and Ranked Pairs Elections: Supplementary Material}

\newcommand{\OMIT}[1]{}

\newcommand{\sproofof}[1]{{\textsc{Proof of {#1}.}}\hspace*{1em}}
\newcommand{\eproofof}[1]{{\hspace*{0.1in} \hfill \qed~{\scriptsize #1}}\bigskip}
\newcommand{\sproofsketch}{{\textit{Proof sketch.}}\hspace*{1em}}
\newcommand{\eproofsketch}{\bigskip}

\newcommand{\threesat}{\ensuremath{\mathrm{3\textsc{SAT}}}}

\newcommand{\rxthreec}{\textsc{RX3C}}

\newcommand{\ppvclong}{\textsc{Path-Preserving Vertex Cut}}

\newcommand{\mippvclong}{\textsc{Maximum-Inclusion Path-Preserving Vertex Cut}}

\newcommand{\cvclong}{\textsc{Colored Path-Preserving Vertex Cut}}

\newcommand{\dcdc}{\textsc{Destructive Control by Deleting Candidates}}
\newcommand{\dcdcg}{\textsc{Destructive Control by Deleting Candidate Groups}}
\newcommand{\dcacg}{\textsc{Destructive Control by Adding Candidate Groups}}

\newcommand{\ccrv}{\textsc{Constructive Control by Replacing Voters}}
\newcommand{\dcrc}{\textsc{Destructive Control by Replacing Candidates}}
\newcommand{\dcrv}{\textsc{Destructive Control by Replacing Voters}}
\newcommand{\schulzeproblem}[1]{Schulze-\textsc{#1}}
\newcommand{\rpproblem}[1]{Ranked-Pairs-\textsc{#1}}

\newcommand{\voteW}[2]{\ensuremath{W{(#1,#2)}}}
\newcommand{\strongestpath}[2]{\ensuremath{P{(#1,#2)}}}

\newcommand{\numberrankedabove}[3]{\ensuremath{N_{#1}{(#2,#3)}}}
\newcommand{\differencerankedabove}[3]{\ensuremath{D_{#1}{(#2,#3)}}}

\newcommand{\pairwise}[3][V]{\ensuremath{N_{#1}{(#2,#3)}}}
\newcommand{\pairwiseDiff}[3][V]{\ensuremath{D_{#1}{(#2,#3)}}}

\newcommand{\str}{\text{str}}
\newcommand{\strPath}[3][V]{\ensuremath{\str_{#1}{(#2,#3)}}}

\newcommand{\lConst}{\ensuremath{\mathcal{L}}}
\newcommand{\lAV}{\ensuremath{\ell_{AV}}}
\newcommand{\lDV}{\ensuremath{\ell_{DV}}}
\newcommand{\lAC}{\ensuremath{\ell_{AC}}}
\newcommand{\lDC}{\ensuremath{\ell_{DC}}}
\newcommand{\lRV}{\ensuremath{\ell_{RV}}}

\newcommand{\lB}{\ensuremath{\ell_{B}}}

\newcommand{\Klasse}[1]{\ensuremath{\mathrm{#1}}}
\newcommand{\p}{\Klasse{P}}
\newcommand{\np}{\Klasse{NP}}

\newtheorem{theorem}{Theorem}
\newtheorem{corollary}{Corollary}
\newtheorem{lemma}{Lemma}
\newtheorem{example}{Example}
\newtheorem{proposition}{Proposition}

\DeclarePairedDelimiter{\cardinality}{\lvert}{\rvert}

\newcommand{\EP}[3]{
	\smallskip
	\begin{center}
		{\small 
			\begin{tabularx}{1.0\columnwidth}{Ll}
				\toprule
				\multicolumn{2}{c}{\sc{#1}} \\
				\midrule
				{\bf Given:}& \parbox[t]{0.79\columnwidth}{#2\vspace*{1mm}} \\
				{\bf Question:}& \parbox[t]{0.79\columnwidth}{#3\vspace*{.5mm}} \\ 
				\bottomrule
			\end{tabularx}
		}
	\end{center}
	\smallskip
}

\newcolumntype{L}{>{\raggedright\arraybackslash}X}
\newcolumntype{R}{>{\raggedleft\arraybackslash}X}
\newcolumntype{C}{>{\centering\arraybackslash}X}

\title{\longTitle}

\author{
Cynthia Maushagen
\and
David Niclaus
\and
Paul Nüsken
\and
Jörg Rothe
\And
Tessa Seeger
\\
\affiliations
Heinrich-Heine-Universität Düsseldorf, MNF, Institut für Informatik, Düsseldorf, Germany\\
\emails
\{cynthia.maushagen, david.niclaus, paul.nuesken, rothe, tessa.seeger\}@hhu.de
}

\newcommand\unnumberedfootnote[1]{%
	\begingroup
	\renewcommand\thefootnote{}\footnote{#1}%
	\addtocounter{footnote}{-1}%
	\endgroup
}

\begin{document}

\maketitle

\begin{abstract}
	Both Schulze and ranked pairs are voting rules that satisfy many natural, desirable axioms. 
	Many standard types of electoral control (with a chair seeking to change the outcome of an election by interfering with the election structure) have already been studied.
	However, for control by replacing candidates or voters and for (exact) multimode control that combines multiple standard attacks, many questions remain open.
	We solve a number of these open cases for Schulze and ranked pairs.
	In addition, we fix a flaw in the reduction of Menton and Singh~[IJCAI~\citeyear{men-sin:c:control-complexity-schulze}] showing that Schulze is resistant to constructive control by deleting candidates and re-establish a vulnerability result for destructive control by deleting candidates.
	In some of our proofs, we study variants of $s$-$t$ vertex cuts in graphs that are related to our control problems.
\end{abstract}

\section{Introduction}\label{sec:introduction}

\unnumberedfootnote{
	This paper is to appear in the proceedings of IJCAI 2024 and in addition contains the supplementary material.
}
Elections play a fundamental role in decision-making processes of societies.
Both the Schulze method~\cite{sch:j:schulze-voting,sch:t:schulze-method-of-voting} and the ranked pairs method~\cite{tid:j:independence-of-clones} satisfy many natural, desirable axioms. %

The Schulze method is a relatively new voting rule and has gained unusual popularity over the past decade due to its outstanding axiomatic properties.
In the real world, organizations like the Wikimedia Foundation, Kubernetes, or the Debian Vote Engine~\cite{sch:t:schulze-method-of-voting} have used it in their decision-making processes. 
Although winner determination with the Schulze method is fairly complicated compared to most other voting rules, it can still be done in polynomial time~\cite{sch:j:schulze-voting,sor-vas-xu:c:fine-grained-complexity-schulze}.
The ranked pairs method was specifically designed to satisfy the \text{independence of clones} criterion~\cite{tid:j:independence-of-clones}.
In general, its axiomatic properties are as outstanding as Schulze's. 
That ranked pairs is barely widespread might be due to the fact that winner determination strongly depends on the handling of ties: When using ``parallel universe tie-breaking''~\cite{con-rog-xia:c:mle}, the winner determination problem is \np-complete. However, it becomes tractable when defining ranked pairs as a resolute rule by using a fixed tie-breaking method~\cite{bri-fis:c:price-of-neutrality}.

We study \emph{electoral control} where a so-called \emph{election chair} (or, simply \emph{chair}) attempts to change the outcome of an election by changing its structure.
Common examples are adding, deleting, partitioning~\cite{bar-tov-tri:j:control,hem-hem-rot:j:destructive-control}, or replacing~\cite{lor-nar-ros-ven-wal:c:replacing-candidates} candidates or voters.
In addition to these control types, we also study multimode control~\cite{fal-hem-hem:j:multimode-control} where the chair can combine several attacks into one. 
For each corresponding electoral control problem, there is a \textit{constructive} case~\cite{bar-tov-tri:j:control} where the goal is to make a favored candidate win the election, and a \textit{destructive} case~\cite{hem-hem-rot:j:destructive-control} where the chair's aim is to prevent a despised candidate from winning. 
It is natural to assume that it is beneficial for a voting rule to be immune to control, i.e., it is impossible to change the outcome of an election by that control type.
However, immunity to control types does not occur often.\footnote{Immunity to control corresponds to strategyproofness against manipulation (a.k.a.\ strategic voting), since both notions formalize the impossibility to reward strategic behaviour with success.
}
In fact, most voting rule are susceptible to control~\cite{fal-rot:b:handbook-comsoc-control-and-bribery}, i.e., control is possible in a least some instances.
Computational intractability can then be seen as a form of resistance to control: If it is computationally hard for an agent to decide if the goal of the attack can be achieved, it may deter the attacker from spending resources on this task. 
On the other hand, not all forms of control are malicious. In many cases, deciding if a control action can be achieved in polynomial time is beneficial for deciding whether to allocate resources to, e.g., a voter drive or spawning new nodes in the context of large clusters. 
We call a voting rule \textit{vulnerable} to a control type if the corresponding decision problem can be decided in polynomial time.

While many standard control types have already been studied for Schulze and ranked pairs (Table~\ref{tab:controlresults}), control by replacing candidates or voters and exact multimode control remained open. 
There are many real world situations where a chair must adhere to a specific number of candidates, e.g., if the election's size is predetermined and a fixed number of candidates are already nominated, forcing the chair to nominate exactly the missing number of candidates. The restriction makes control impossible in some situations, where control by adding fewer candidates would be successful.
Other (practical) settings include autonomous agents where, e.g., the size of the cluster is predetermined. A chair may be able to influence some part of the cluster, but is forced to adhere to the overall determined size (examples can be found in the supplementary material).

\paragraph{Related Work:}
Bartholdi et al.~\shortcite{bar-tov-tri:j:control} introduced constructive control types and Hemaspaandra et al.~\shortcite{hem-hem-rot:j:destructive-control} the corresponding destructive cases.
Control by replacing candidates or voters was introduced by Loreggia et al.~\shortcite{lor-nar-ros-ven-wal:c:replacing-candidates}, while Faliszewski et al.~\shortcite{fal-hem-hem:j:multimode-control} introduced and studied multimode control. 
Erd{\'e}lyi et al.~\shortcite{erd-nev-reg-rot-yan-zor:j:towards-completing-the-puzzle} provide an extensive study and overview of various control problems, including replacing candidates or voters and also
multimode control.\footnote{Besides for voting, control has also been studied in, e.g., judgment aggregation \cite{bau-erd-erd-rot-sel:j:complexity-of-control-in-judgment-aggregation-for-uniform-premise-based-quota-rules} and for weighted voting games \cite{rey-rot:j:structural-control-in-weighted-voting-games,kac-rot:j:controlling-weighted-voting-games-by-deleting-or-adding-players-with-or-without-changing-the-quota} and graph-restricted weighted voting games \cite{kac-rot-tal:c:complexity-of-control-by-adding-or-deleting-edges-in-graph-restricted-weighted-voting-games}.}

Other types of strategic influence on elections are \emph{manipulation}, where a voter or a group of voters state their preferences strategically (i.e., untruthfully), and \emph{bribery}, where a controlling agent bribes voters to change their preferences  (see, e.g., \cite{bar-orl:j:polsci:strategic-voting,bar-tov-tri:j:manipulating,con-san:c:nonexistence,fal-hem-hem:j:bribery,fal-hem-hem-rot:j:llull-copeland-full-techreport}).
Control, bribery, and manipulation have been studied for a wide range of voting rules, as surveyed by Faliszwski and Rothe~\shortcite{fal-rot:b:handbook-comsoc-control-and-bribery} (bribery and control) and Conitzer and Walsh~\shortcite{con-wal:b:handbook-comsoc-manipulation} (manipulation), see also \cite{bau-rot:b:economics-and-computation-preference-aggregation-by-voting}.

For Schulze and ranked pairs, Parkes and Xia~\shortcite{par-xia:c:strategic-schulze-ranked-pairs}, Xia et al.~\shortcite{xia-zuc-pro-con-ros:c:unweighted-coalitional-manipulation}, Menton and Singh~\shortcite{men-sin:c:control-complexity-schulze}, and Gaspers et al.~\shortcite{gas-ka-nar-wal:c:coalitional-manipulation-schulze} studied constructive and destructive control by adding or deleting voters or candidates, bribery, and manipulation.
Table~\ref{tab:controlresults} gives an overview of known results.
Hemaspaandra et al.~\shortcite{hem-lav-men:c:schulze-ranked-pairs-fpt} showed fixed-parameter tractability for bribing, controlling, and manipulating Schulze and ranked pairs elections with respect to the number of candidates and provided algorithms with uniform polynomial running time that are independent of the number of candidates. 
Menton and Singh~\shortcite{men-sin:c:control-complexity-schulze} also provided results on control by partition and runoff partition of candidates and partition of voters for Schulze and further showed some results for all Condorcet-consistent voting rules.

\begin{table}
	\caption[Overview of complexity results for standard control (AC, DC, AV, DV), bribery (B), and manipulation (M) in Schulze and ranked pairs elections.]{%
		Overview of complexity results for standard control (AC, DC, AV, DV), bribery (B), and manipulation (M) in Schulze and ranked pairs elections.
		Our results are in \textcolor{blue}{\bf blue}.
		Results marked by $\spadesuit$ are due to Parkes and Xia~\shortcite{par-xia:c:strategic-schulze-ranked-pairs};
		by $\clubsuit$ due to Xia et al.~\shortcite{xia-zuc-pro-con-ros:c:unweighted-coalitional-manipulation}; and 
		by $\blacklozenge$
		claimed to be in $\p$ by Menton and Singh~\shortcite{men-sin:t:manipulation-control-schulze-voting-v1}, but later stated as open~\shortcite{men-sin:c:control-complexity-schulze} (and omitted from their most recent arXiv version, v4, dated May 24, 2013), and
		re-established in Theorem~\ref{thm:schulze-dcdc-our-result}.
		$\bigstar$ marks a result by Parkes and Xia~\shortcite{par-xia:c:strategic-schulze-ranked-pairs}, extended by Gaspers et al.~\shortcite{gas-ka-nar-wal:c:coalitional-manipulation-schulze}.
		The original proof of Menton and Singh~\shortcite{men-sin:c:control-complexity-schulze} for the result marked by $\diamondsuit$ is corrected in Section~\ref{sec:control-ccdc} and extended to the unique-winner model in Theorem ~\ref{thm:schulze-ccdc-unique}.
		All results, except Schulze-\textsc{DCDC} (where the unique-winner model remains open), hold in both winner models.
	}
	\label{tab:controlresults}
	\begin{tabularx}{\columnwidth}{@{}l@{\hspace*{1mm}}|@{\hspace*{1mm}}C@{\hspace*{1mm}}@{\hspace*{1mm}}C@{\hspace*{1mm}}|@{\hspace*{1mm}}C@{\hspace*{1mm}}@{\hspace*{1mm}}C@{\hspace*{1mm}}}
		\toprule
		& \multicolumn{2}{@{\hspace*{1mm}}c|@{\hspace*{1mm}}}{Schulze} &  \multicolumn{2}{c}{Ranked pairs} \\
		& Constructive & Destructive & Constructive & Destructive \\
		\midrule
		AC & NP-c.$^\spadesuit$ &  ? &  NP-c.$^\spadesuit$ &  NP-c.$^\spadesuit$ \\
		DC & \textcolor{blue}{\bf NP-c.}$^{\diamondsuit}$ & \textcolor{blue}{\bf P}$^\blacklozenge$ &  NP-c.$^\spadesuit$ &  NP-c.$^\spadesuit$ \\
		AV & NP-c.$^\spadesuit$ & NP-c.$^\spadesuit$ &  NP-c.$^\spadesuit$ &  NP-c.$^\spadesuit$ \\
		DV & NP-c.$^\spadesuit$ & NP-c.$^\spadesuit$ &  NP-c.$^\spadesuit$ & NP-c.$^\spadesuit$  \\
		B & NP-c.$^\spadesuit$ & NP-c.$^\spadesuit$ &  NP-c.$^\spadesuit$ &  NP-c.$^\spadesuit$ \\
		M & P$^{\bigstar}$ & P$^\spadesuit$ & NP-c.$^\clubsuit$ & NP-c.$^\spadesuit$ \\
		\bottomrule		
	\end{tabularx}
\end{table}

\section{Preliminaries}
\label{sec:prelim}

An \textit{election} is a pair $(C,V)$, where $C= \{c_1, \ldots, c_m\}$ is a set of \textit{candidates} and $V$ is a list of $n$ votes in form of \textit{preferences} over all candidates in~$C$. 
In this paper, preferences are expressed as a strict linear order over~$C$,\footnote{Some voting rule (notably, e.g., Schulze voting~\shortcite{sch:j:schulze-voting}) allow voters to express preferences as weak orders. Following previous work (e.g., \cite{men-sin:c:control-complexity-schulze,par-xia:c:strategic-schulze-ranked-pairs,hem-lav-men:c:schulze-ranked-pairs-fpt}), we restrict
	them here to linear orders.} where voters rank the candidates in descending order from most to least preferred.
We write $c\succ_{v_i} d$ to express that a voter $v_i \in V$ prefers candidate $c$ over~$d$.
When it is clear from the context, we omit $\succ_{v_i}$ and simply write $c\, d$.
For a set of candidates $A \subseteq C$, we refer to the lexicographic order (we first consider the alphabetic order and then the subscript) of those candidates by simply writing~$\overrightarrow{A}$, and $\overleftarrow{A}$ for the reverse.
To shorten notation of votes we will sometimes include sets of candidates in the votes.
If a set $A$ occurs in a vote, the candidates from $A$ are ranked in lexicographic order at the given place of the voters ranking.
For example, $c\,A$ is the same as $c\,\overrightarrow{A}$ and denotes the vote that ranks $c$ first and then all candidates from $A$ follow in increasing lexicographic order.
For a set of candidates $C$ and two candidates $c,d \in C$, we will write $\voteW{c}{d}$ for the two votes $c \, d \, \overrightarrow{C\setminus \{c,d\}}$ and $\overleftarrow{C\setminus \{c,d\} } \, c \, d$.
For a given election $(C,V)$, let $\numberrankedabove{V}{c}{d}$ be the number of votes, in which candidate $c$ is ranked above candidate~$d$. We call this the \textit{pairwise comparison} between $c$ and $d$. Similarly, let $\differencerankedabove{V}{c}{d} = \numberrankedabove{V}{c}{d} - \numberrankedabove{V}{d}{c}$.
Note that we omit the list of votes if it is clear from the context.
Before we introduce the procedures for winner determination in Schulze and ranked pairs election, we first introduce the \textit{weighted majority graph (WMG)}.
A WMG for an election $(C,V)$ is a weighted directed graph $G = (\hat{V},E,w)$, where $\hat{V} = C$ and $(c,d) \in E$ with weight $w(c,d) = \differencerankedabove{V}{c}{d}$ for each ordered pair of candidates $c,d \in C$. Note that, since the votes are expressed as strict orders, we have $\differencerankedabove{V}{c}{d}  = - \differencerankedabove{V}{d}{c}$. Therefore, we omit edges with a negative weight in depictions. %

A voting rule $r:\{(C,V)\mid (C,V) \text{ is an election}\} \rightarrow 2^C$ determines the set of winners of an election $(C,V)$. We focus on the voting rules \textit{Schulze} and \textit{ranked pairs}.
A candidate $c \in C$ is a \textit{Condorcet winner} of an election $(C,V)$ if $\numberrankedabove{V}{c}{d} > \numberrankedabove{V}{d}{c}$ holds for all $d \in C$ with $c \neq d$.
A candidate $c \in C$ is a \textit{weak Condorcet winner} if $\numberrankedabove{V}{c}{d} \geq \numberrankedabove{V}{d}{c}$ for all $d \in C$ with $c \neq d$.
Note that there can be at most one Condorcet winner but possibly multiple weak Condorcet winners.
Both Schulze and ranked pairs are \emph{Condorcet-consistent voting rules}, meaning they choose the Condorcet winner whenever there is one.

\paragraph{Schulze:} Let the \textit{strength of a path $p$} ($\str(p)$) be the weight of the weakest edge, i.e., the minimum weight, in a directed path~$p$ between two candidates in the WMG. 
For each distinct pair of candidates $c,d \in C$, let the \textit{strength of the strongest path} be
$P(c,d) = \max \{ \str(p) \mid \text{$p$ is a path from $c$ to $d$}\}$.
A candidate $c \in C$ is a \textit{Schulze winner of $(C,V)$} if $P(c,d) \geq P(d,c)$ for each $d \in C\setminus \{c\}$. To find these candidates one can build a second directed graph, which has an edge from $c$ to $d$ if and only if $P(c,d) > P(d,c)$.
Each candidate with an in-degree of zero in this graph is a \emph{Schulze winner}. %
There can be multiple Schulze winners.

\paragraph{Ranked pairs:}
For an election $(C,V)$, we first calculate $\differencerankedabove{V}{c}{d}$ for all distinct pairs of candidates $c,d \in C$, %
and order the pairs by weight from highest to lowest, i.e., we order the values of $\differencerankedabove{V}{c}{d}$. 
Now in each step, we consider the top pair $(c,d)$ of this weight order, which has not yet been considered. 
Following Parkes and Xia~\shortcite{par-xia:c:strategic-schulze-ranked-pairs}, we break ties according to a fixed tie-breaking rule.  
We add an edge $(c,d)$ to a directed graph~$G = (V',E)$ where $V'=C$, unless inserting this edge would create a cycle, in which case the pair (edge) is disregarded.
When all pairs have been considered, the \emph{ranked pairs winner of $(C,V)$} (subject to the fixed tie-breaking) is the candidate corresponding to the source of G.\footnote{In the original definition, the voting rule ranked pairs~\cite{tid:j:independence-of-clones}
	returns a complete ranking of the candidates.
	We use a slightly simplified
	definition of ranked pairs introduced by Berker et al.~\shortcite{ber-cas-ong-rob:t:obvious-independence-of-clones} that only returns the winner of the election. 
}

\smallskip

For examples of the Schulze and ranked pairs method we refer to the supplementary material.
Note that Tideman~\shortcite{tid:j:independence-of-clones} originally gives an irresolute procedure and corresponding function for ranked pairs, which considers all possible ways of breaking ties and is \np-complete~\cite{bri-fis:c:price-of-neutrality}. 
Interestingly, the full ranking of candidates and winner determination including ties can be computed in polynomial time for Schulze~\shortcite{sch:t:schulze-method-of-voting}.
For our reductions, it is important that the procedures run in polynomial time (which we achieve by using a fixed tie-breaking scheme for ranked pairs), but it is not necessary to be highly efficient.
Consequently, there may be other procedures that %
are faster (see, e.g., \cite{sor-vas-xu:c:fine-grained-complexity-schulze} for Schulze) or handle ties in a different way for ranked pairs (see, e.g., \cite{bri-fis:c:price-of-neutrality,wan-sik-sha-zha-jia-xia:c:algorithms-multi-stage-voting-put}).

We study various types of electoral control, starting with constructive control by deleting candidates (CCDC) which was defined by Bartholdi et al.~\shortcite{bar-tov-tri:j:control} for a voting rule $\mathcal{E}$: %

\EP{$\mathcal{E}$-\textsc{Constructive Control By Deleting Candidates}}
{An election $(C,V)$, a distinguished candidate $p \in C$, and $\ell \in \mathbb{N}$.}
{Is it possible to make $p$ the unique winner of the $\mathcal{E}$ election resulting from $(C,V)$ by deleting at most $\ell$ candidates?}

In the setting of replacing candidates or voters~\cite{lor-nar-ros-ven-wal:c:replacing-candidates,erd-nev-reg-rot-yan-zor:j:towards-completing-the-puzzle}, the chair must not alter the size of the election and instead must add a candidate or voter for each one she deletes.
Formally, in $\mathcal{E}$-\textsc{Constructive Control By Replacing Candidates} ($\mathcal{E}$-\textsc{CCRC}) we are given %
two disjoint sets of candidates, $C$ and~$D$, a list of votes over $C \cup D$, a distinguished candidate $p \in C$, and $\ell \in \mathbb{N}$, and
we ask if it is possible to make $p$ the unique winner of the $\mathcal{E}$ election resulting from $(C,V)$ by replacing at most $\ell$ candidates $C' \subseteq C$ with candidates $D' \subseteq D$, where $\vert C' \vert = \vert D' \vert$.
$\mathcal{E}$-\ccrv\ ($\mathcal{E}$-\textsc{CCRV}) is defined analogously by asking whether it is possible to make a preferred candidate $p$ the unique winner by replacing at most $\ell$ votes $V' \subseteq V$ with votes $U' \subseteq U$ such that $\vert V' \vert = \vert U' \vert$, where $U$ is a list of as yet unregistered votes.

In these control scenarios, a chair's goal is to make a preferred candidate the unique winner.
A chair may also be interested in preventing a candidate from winning.
This setting is known as destructive control~\cite{hem-hem-rot:j:destructive-control}.
Instead of asking whether a candidate can be made the winner, we ask whether the sole victory of a candidate can be prevented (note that for a control action to be successful, it is enough to have the candidate be a winner among others).
We write $\mathcal{E}$-\textsc{DCDC} for $\mathcal{E}$-\dcdc, $\mathcal{E}$-\textsc{DCRC} for $\mathcal{E}$-\dcrc, and $\mathcal{E}$-\textsc{DCRV} for $\mathcal{E}$-\dcrv.
Aside from the \textit{unique-winner model} which is used in the previous definitions, in the \textit{nonunique-winner model} we ask whether a preferred candidate can be made a winner (possibly among others) in the constructive case, and whether a despised candidate can be prevented from winning altogether in the destructive case.
Note that when interpreting a voting rule as \emph{resolute} (i.e., to always yield exactly one winner), the unique-winner and nonunique-winner models are the same.

In addition to the above control problems, we also study variations of multimode control as introduced by Faliszewski et al.~\shortcite{fal-hem-hem:j:multimode-control}, where several of the standard control attacks and bribery can be combined into one action:

\EP{$\mathcal{E}$-\textsc{Constructive Control By AC+DC+AV+DV+B}}
{Two disjoint sets of candidates, $C$ and~$D$, two disjoint lists of votes over $C \cup D$, $V$ and $U$, a distinguished candidate $p \in C$, and $\lAC$, $\lDC$, $\lAV$, $\lDV$, $\lB \in \mathbb{N}$.}
{Is it possible to find two sets, $C' \subseteq C \setminus \{p\}$ and $D' \subseteq D$, and two sublists of votes, $V' \subseteq V$ and $U' \subseteq U$, such that $p$ is the unique winner of the $\mathcal{E}$ election that results from $((C \setminus C') \cup D', (V \setminus V') \cup U')$ by bribing at most $\lB$ votes in $(V \setminus V') \cup U')$, and
	$\vert D' \vert \leq \lAC$,
	$\vert C' \vert \leq \lDC$,
	$\vert U' \vert \leq \lAV$, and
	$\vert V' \vert \leq \lDV$?}

We abbreviate multimode control problems in the obvious way; e.g., we use the shorthand $\mathcal{E}$-\textsc{CCAC+DC+AV+DV+B} for the above problem.
Faliszewski et al.~\shortcite{fal-hem-hem:j:multimode-control} define a method to classify all $2^5-1$ variants of multimode control for a voting rule, called classification rule A.
Using the classification rule A and the known results for adding and deleting candidates or voters and for bribery (see Table~\ref{tab:controlresults}), it immediately follows that, except for Schulze-DCAC+DC, Schulze is resistant to any multimode attack.
Since ranked pairs is resistant to all single-pronged attacks, it clearly also resists all combinations of multimode control.

In any instance of $\mathcal{E}$-\textsc{Exact Constructive Control by AC+DC+AV+DV+B},
it must hold that $\vert D' \vert = \lAC$, $\vert C' \vert = \lDC$, $\vert U' \vert = \lAV$,  $\vert V' \vert = \lDV$, and exactly $\lB$ voters in $(V \setminus V') \cup U'$ are bribed. Each corresponding nonexact control problem polynomial-time Turing reduces to the exact control problem.
The destructive variants $\mathcal{E}$-\textsc{Destructive Control by AC+DC+AV+DV+B} and $\mathcal{E}$-\textsc{Exact Destructive Control by AC+DC+AV+DV+B} are defined analogously by asking whether it is possible to make $p$ not a (unique) winner, and we again use the obvious shorthands.
Sometimes, we exclude certain %
actions from multimode control, considering, e.g., only candidate control ($\mathcal{E}$-\textsc{DCAC+DC}) and
omit the unneeded input parameters. %

Note that we do not allow candidates in $D$ or voters in $U$ to be deleted, as we do not consider it realistic to remove a candidate or voter from an election right after adding.
Even though this does not make a difference for the nonexact problems (as it is allowed to simply add or delete fewer candidates or voters), it may affect the result for exact multimode control as shown in the supplementary material.

\section{Schulze Resists Constructive Control by Deleting Candidates}
\label{sec:control-ccdc}

In this section, we prove the following result by describing and fixing a flaw in the proof of Menton and Singh~\shortcite{men-sin:c:control-complexity-schulze}.\footnote{Menton and Singh agree that our construction fixes their flaw.
	Menton writes (private email communication on December 5, 2023), \emph{``I have tried my best to understand the old proof and your correction of it.
		The diagrams in your paper were very helpful for this!
		[\ldots] Your solution changes the votes so there are lower weight edges in the graph and additional paths from $p$ back to the clause candidates such that there is an equal weight path from $p$ to the clause candidates to those persisting paths.
		Deleting the appropriate literal candidates is still necessary to break the higher-weight paths from the clause candidates to~$p$.
		Looks good to me!''}}

\begin{theorem}\label{thm:schulze-ccdc}
	Schulze-\textsc{CCDC}\ is \np-complete in the nonunique-winner model.
\end{theorem}

\begin{proof}
	The proof of this result, due to Menton and Singh~\shortcite[Thm.~2.2]{men-sin:c:control-complexity-schulze}, shows a clever reduction from $\threesat$, but it is technically flawed.
	We briefly present their reduction %
	and give a counterexample showing that it is not correct. %
	
	In the \textsc{3-Satisfiability} problem ($\threesat$), we are given a set $X$ of variables and a set $Cl = \{Cl_1, \ldots, Cl_k\}$ of clauses over~$X$, each having exactly three literals, and we ask whether there is a satisfying assignment for~$\varphi$, where $\varphi$ is the conjunction of all clauses $Cl_i \in Cl$.
	Given a $\threesat$ instance $(X, Cl)$, Menton and Singh~\shortcite{men-sin:c:control-complexity-schulze} construct a Schulze-\textsc{CCDC} instance $((C,V'),p,k)$ as follows.
	The set of candidates $C$ contains 
		$k+1$ \textit{clause candidates} $c_i^1, \dots , c_i^{k+1}$ for 
		each clause $Cl_i \in Cl$,
		three \textit{literal candidates} $x_i^1, x_i^2, x_i^3$ for each  
		clause~$Cl_i$, where $x_i^j$ is the \textit{j}th literal in clause~$Cl_i$, 
		$k+1$ \textit{negation candidates} $n^1_{i,j,m,n}, \dots ,
		n^{k+1}_{i,j,m,n}$ for each pair of literals $x_i^j, x_m^n$, where one 
		is the negation of the other, and
		the distinguished candidate $p$ and an additional candidate~$a$. 
	Let $C_i = \{ c_i^1, \dots , c_i^{k+1}\}$ be the set of all clause candidates for clause $Cl_i \in Cl$ and let $K=\bigcup_{i=1}^{k+1} C_i$ be the set of all clause candidates.
	Let $L_i = \{x_i^1, x_i^2, x_i^3\}$ be the set of literal candidates for the clause $Cl_i$ and let $L = \bigcup_{i=1}^{k} L_i$ be the set of all literal candidates.
	Let $N_{ijmn} = \{ n_{ijmn}^1, \dots , \ n_{ijmn}^{k+1}  \}$ be the set of negation candidates for the literals $x_i^j$, $x_m^n$ that are a negation of each other, and let $N$ be the set of all such negation candidates.
	For a positive integer~$z$, we write $[z] = \{1, \ldots , z\}$ as a shorthand.
	
	Menton and Singh~\shortcite{men-sin:c:control-complexity-schulze} define the following list of votes $V'$ (which we will change later to fix the proof):
	
	{
		\vspace{1em}
		\noindent
		\begin{tabularx}{\columnwidth}{cclL}
			&$\#$ & preferences & for each \\
			\midrule
			(A)&$1$ & $\voteW{c_i^j}{x_i^1}$ & $i\in [k]$, $j\in [k+1]$\\
			(B)&$1$ & $\voteW{x_i^1}{x_i^2}$ & $i\in [k]$\\
			(C)&$1$ & $\voteW{x_i^2}{x_i^3}$ & $i\in [k]$\\
			(D)&$1$ & $\voteW{x_i^3}{p}$ & $i\in [k]$\\
			(E)&$1$ & $\voteW{a}{x}$ & $x\in L$\\
			(F)&$1$ & $\voteW{p}{a}$ & \\
			(G)&$1$ & $\voteW{x_i^j}{n^l_{ijmn}}$ & $l \in [k+1]$ where $x_m^n$ is the negation of $x_i^j$\\
			(H)&$1$ & $\voteW{n}{p}$ &  $n\in N$\\ 
		\end{tabularx}
	}
	
	The deletion limit is~$k$, the number of clauses. 
	Menton and Singh~\shortcite{men-sin:c:control-complexity-schulze} argue that $p$ can be made a Schulze winner by deleting at most $k$ candidates from $C$ if and only if there
	is a truth assignment that makes the given $\threesat$ instance true.
	
	We now briefly present our counterexample, where we map a yes-instance of $\threesat$ to a no-instance of Schulze-\textsc{CCDC}.
	Let $(X,Cl)$ be our given $\threesat$ instance, with $X=\{x_1,x_2,x_3\}$ and $Cl=\{(x_1\vee x_2 \vee \neg {x}_3), \left(\neg {x}_{1} \vee x_{2} \vee x_3\right)\}$, i.e., we consider the CNF formula
	\[
	\varphi = \left( x_{1} \vee x_{2} \vee \neg {x}_{3}\right) \wedge
	\left(\neg {x}_{1} \vee x_{2} \vee x_3\right).
	\]
	A detailed description of this counterexample is given in the supplementary material.
	Their reduction is quite clever, but unfortunately wrong, as shown by the counterexample.
	However, by modifying it appropriately, we can ensure that $p$ can indeed be made a Schulze winner of the election by deleting at most $k$ candidates if and only if $(X,Cl)$ is a yes-instance of $\threesat$.
	For our modifications, it is only necessary to change the list of votes. For votes (A) through (F) we adjust the number of votes to $2$. Additionally, we add $1$ vote $W(a,c)$ for each $c \in K$.
	\begin{figure}
		\begin{center} 
			\tikzset{every picture/.style={line width=0.75pt}} %
			\vspace*{-1.4cm}
			\hspace*{-5em}
			\small\begin{tikzpicture}[x=0.43pt,y=0.43pt,yscale=-1,xscale=1]
				\draw   (53,45) .. controls (53,37.82) and (58.82,32) .. (66,32) .. controls (73.18,32) and (79,37.82) .. (79,45) .. controls (79,52.18) and (73.18,58) .. (66,58) .. controls (58.82,58) and (53,52.18) .. (53,45) -- cycle ;
				\draw   (53,85) .. controls (53,77.82) and (58.82,72) .. (66,72) .. controls (73.18,72) and (79,77.82) .. (79,85) .. controls (79,92.18) and (73.18,98) .. (66,98) .. controls (58.82,98) and (53,92.18) .. (53,85) -- cycle ;
				\draw   (53,125) .. controls (53,117.82) and (58.82,112) .. (66,112) .. controls (73.18,112) and (79,117.82) .. (79,125) .. controls (79,132.18) and (73.18,138) .. (66,138) .. controls (58.82,138) and (53,132.18) .. (53,125) -- cycle ;
				\draw   (53,222) .. controls (53,214.82) and (58.82,209) .. (66,209) .. controls (73.18,209) and (79,214.82) .. (79,222) .. controls (79,229.18) and (73.18,235) .. (66,235) .. controls (58.82,235) and (53,229.18) .. (53,222) -- cycle ;
				\draw   (53,262) .. controls (53,254.82) and (58.82,249) .. (66,249) .. controls (73.18,249) and (79,254.82) .. (79,262) .. controls (79,269.18) and (73.18,275) .. (66,275) .. controls (58.82,275) and (53,269.18) .. (53,262) -- cycle ;
				\draw   (53,302) .. controls (53,294.82) and (58.82,289) .. (66,289) .. controls (73.18,289) and (79,294.82) .. (79,302) .. controls (79,309.18) and (73.18,315) .. (66,315) .. controls (58.82,315) and (53,309.18) .. (53,302) -- cycle ;
				\draw   (325,72.01) .. controls (332.18,72.01) and (338,77.83) .. (338,85.01) .. controls (338,92.19) and (332.18,98.01) .. (325,98.01) .. controls (317.82,98.01) and (312,92.18) .. (312,85) .. controls (312,77.83) and (317.82,72.01) .. (325,72.01) -- cycle ;
				\draw   (245,72) .. controls (252.18,72) and (258,77.82) .. (258,85) .. controls (258,92.18) and (252.18,98) .. (245,98) .. controls (237.82,98) and (232,92.18) .. (232,85) .. controls (232,77.82) and (237.82,72) .. (245,72) -- cycle ;
				\draw   (163,71.99) .. controls (170.18,71.99) and (176,77.82) .. (176,85) .. controls (176,92.17) and (170.18,97.99) .. (163,97.99) .. controls (155.82,97.99) and (150,92.17) .. (150,84.99) .. controls (150,77.81) and (155.82,71.99) .. (163,71.99) -- cycle ;
				\draw   (325,250.01) .. controls (332.18,250.01) and (338,255.83) .. (338,263.01) .. controls (338,270.19) and (332.18,276.01) .. (325,276.01) .. controls (317.82,276.01) and (312,270.18) .. (312,263) .. controls (312,255.83) and (317.82,250.01) .. (325,250.01) -- cycle ;
				\draw   (245,250) .. controls (252.18,250) and (258,255.82) .. (258,263) .. controls (258,270.18) and (252.18,276) .. (245,276) .. controls (237.82,276) and (232,270.18) .. (232,263) .. controls (232,255.82) and (237.82,250) .. (245,250) -- cycle ;
				\draw   (163,248.99) .. controls (170.18,248.99) and (176,254.82) .. (176,262) .. controls (176,269.17) and (170.18,274.99) .. (163,274.99) .. controls (155.82,274.99) and (150,269.17) .. (150,261.99) .. controls (150,254.81) and (155.82,248.99) .. (163,248.99) -- cycle ;
				\draw    (79,45) -- (148.26,84.01) ;
				\draw [shift={(150,84.99)}, rotate = 209.39] [fill={rgb, 255:red, 0; green, 0; blue, 0 }  ][line width=0.08]  [draw opacity=0] (12,-3) -- (0,0) -- (12,3) -- cycle    ;
				\draw    (79,85) -- (148,84.99) ;
				\draw [shift={(150,84.99)}, rotate = 179.99] [fill={rgb, 255:red, 0; green, 0; blue, 0 }  ][line width=0.08]  [draw opacity=0] (12,-3) -- (0,0) -- (12,3) -- cycle    ;
				\draw    (79,125) -- (148.26,85.97) ;
				\draw [shift={(150,84.99)}, rotate = 150.6] [fill={rgb, 255:red, 0; green, 0; blue, 0 }  ][line width=0.08]  [draw opacity=0] (12,-3) -- (0,0) -- (12,3) -- cycle    ;
				\draw    (79,222) -- (148.26,261.01) ;
				\draw [shift={(150,261.99)}, rotate = 209.39] [fill={rgb, 255:red, 0; green, 0; blue, 0 }  ][line width=0.08]  [draw opacity=0] (12,-3) -- (0,0) -- (12,3) -- cycle    ;
				\draw    (79,262) -- (148,261.99) ;
				\draw [shift={(150,261.99)}, rotate = 179.99] [fill={rgb, 255:red, 0; green, 0; blue, 0 }  ][line width=0.08]  [draw opacity=0] (12,-3) -- (0,0) -- (12,3) -- cycle    ;
				\draw    (80,301) -- (148.25,262.96) ;
				\draw [shift={(150,261.99)}, rotate = 150.87] [fill={rgb, 255:red, 0; green, 0; blue, 0 }  ][line width=0.08]  [draw opacity=0] (12,-3) -- (0,0) -- (12,3) -- cycle    ;
				\draw   (203.76,158.28) .. controls (210.94,158.15) and (216.86,163.86) .. (216.99,171.04) .. controls (217.12,178.22) and (211.41,184.14) .. (204.23,184.27) .. controls (197.05,184.4) and (191.13,178.69) .. (191,171.51) .. controls (190.87,164.33) and (196.58,158.41) .. (203.76,158.28) -- cycle ;
				\draw   (163.76,159) .. controls (170.94,158.87) and (176.87,164.59) .. (177,171.76) .. controls (177.13,178.94) and (171.41,184.87) .. (164.24,185) .. controls (157.06,185.13) and (151.13,179.41) .. (151,172.24) .. controls (150.87,165.06) and (156.59,159.13) .. (163.76,159) -- cycle ;
				\draw   (123.77,159.73) .. controls (130.95,159.6) and (136.87,165.31) .. (137,172.49) .. controls (137.13,179.67) and (131.42,185.59) .. (124.24,185.72) .. controls (117.06,185.85) and (111.14,180.14) .. (111.01,172.96) .. controls (110.88,165.78) and (116.59,159.86) .. (123.77,159.73) -- cycle ;
				\draw   (364.76,156.28) .. controls (371.94,156.15) and (377.86,161.86) .. (377.99,169.04) .. controls (378.12,176.22) and (372.41,182.14) .. (365.23,182.27) .. controls (358.05,182.4) and (352.13,176.69) .. (352,169.51) .. controls (351.87,162.33) and (357.58,156.41) .. (364.76,156.28) -- cycle ;
				\draw   (324.76,157) .. controls (331.94,156.87) and (337.87,162.59) .. (338,169.76) .. controls (338.13,176.94) and (332.41,182.87) .. (325.24,183) .. controls (318.06,183.13) and (312.13,177.41) .. (312,170.24) .. controls (311.87,163.06) and (317.59,157.13) .. (324.76,157) -- cycle ;
				\draw   (284.77,157.73) .. controls (291.95,157.6) and (297.87,163.31) .. (298,170.49) .. controls (298.13,177.67) and (292.42,183.59) .. (285.24,183.72) .. controls (278.06,183.85) and (272.14,178.14) .. (272.01,170.96) .. controls (271.88,163.78) and (277.59,157.86) .. (284.77,157.73) -- cycle ;
				\draw[dashed]    (163,97.99) -- (124.84,158.04) ;
				\draw [shift={(123.77,159.73)}, rotate = 302.43] [fill={rgb, 255:red, 0; green, 0; blue, 0 }  ][line width=0.08]  [draw opacity=0] (12,-3) -- (0,0) -- (12,3) -- cycle    ;
				\draw[dashed]    (163,97.99) -- (163.74,157) ;
				\draw [shift={(163.76,159)}, rotate = 269.28] [fill={rgb, 255:red, 0; green, 0; blue, 0 }  ][line width=0.08]  [draw opacity=0] (12,-3) -- (0,0) -- (12,3) -- cycle    ;
				\draw[dashed]    (163,97.99) -- (202.64,156.62) ;
				\draw [shift={(203.76,158.28)}, rotate = 235.94] [fill={rgb, 255:red, 0; green, 0; blue, 0 }  ][line width=0.08]  [draw opacity=0] (12,-3) -- (0,0) -- (12,3) -- cycle    ;
				\draw[dashed]    (163,248.99) -- (125.29,187.43) ;
				\draw [shift={(124.24,185.72)}, rotate = 58.51] [fill={rgb, 255:red, 0; green, 0; blue, 0 }  ][line width=0.08]  [draw opacity=0] (12,-3) -- (0,0) -- (12,3) -- cycle    ;
				\draw[dashed]    (163,248.99) -- (164.2,187) ;
				\draw [shift={(164.24,185)}, rotate = 91.1] [fill={rgb, 255:red, 0; green, 0; blue, 0 }  ][line width=0.08]  [draw opacity=0] (12,-3) -- (0,0) -- (12,3) -- cycle    ;
				\draw[dashed]    (163,248.99) -- (203.15,185.96) ;
				\draw [shift={(204.23,184.27)}, rotate = 122.5] [fill={rgb, 255:red, 0; green, 0; blue, 0 }  ][line width=0.08]  [draw opacity=0] (12,-3) -- (0,0) -- (12,3) -- cycle    ;
				\draw[dashed]    (325,97.99) -- (285.89,156.07) ;
				\draw [shift={(284.77,157.73)}, rotate = 303.96] [fill={rgb, 255:red, 0; green, 0; blue, 0 }  ][line width=0.08]  [draw opacity=0] (12,-3) -- (0,0) -- (12,3) -- cycle    ;
				\draw[dashed]    (325,97.99) -- (324.77,155) ;
				\draw [shift={(324.76,157)}, rotate = 270.23] [fill={rgb, 255:red, 0; green, 0; blue, 0 }  ][line width=0.08]  [draw opacity=0] (12,-3) -- (0,0) -- (12,3) -- cycle    ;
				\draw[dashed]    (325,97.99) -- (363.63,154.62) ;
				\draw [shift={(364.76,156.28)}, rotate = 235.7] [fill={rgb, 255:red, 0; green, 0; blue, 0 }  ][line width=0.08]  [draw opacity=0] (12,-3) -- (0,0) -- (12,3) -- cycle    ;
				\draw[dashed]    (326,249.99) -- (286.29,185.43) ;
				\draw [shift={(285.24,183.72)}, rotate = 58.41] [fill={rgb, 255:red, 0; green, 0; blue, 0 }  ][line width=0.08]  [draw opacity=0] (12,-3) -- (0,0) -- (12,3) -- cycle    ;
				\draw[dashed]    (326,249.99) -- (325.26,185) ;
				\draw [shift={(325.24,183)}, rotate = 89.34] [fill={rgb, 255:red, 0; green, 0; blue, 0 }  ][line width=0.08]  [draw opacity=0] (12,-3) -- (0,0) -- (12,3) -- cycle    ;
				\draw[dashed]    (325,250.01) -- (364.21,183.99) ;
				\draw [shift={(365.23,182.27)}, rotate = 120.71] [fill={rgb, 255:red, 0; green, 0; blue, 0 }  ][line width=0.08]  [draw opacity=0] (12,-3) -- (0,0) -- (12,3) -- cycle    ;
				\draw    (176,85) -- (230,85) ;
				\draw [shift={(232,85)}, rotate = 180] [fill={rgb, 255:red, 0; green, 0; blue, 0 }  ][line width=0.08]  [draw opacity=0] (12,-3) -- (0,0) -- (12,3) -- cycle    ;
				\draw    (258,85) -- (310,85) ;
				\draw [shift={(312,85)}, rotate = 180] [fill={rgb, 255:red, 0; green, 0; blue, 0 }  ][line width=0.08]  [draw opacity=0] (12,-3) -- (0,0) -- (12,3) -- cycle    ;
				\draw    (176,263) -- (230,263) ;
				\draw [shift={(232,263)}, rotate = 180] [fill={rgb, 255:red, 0; green, 0; blue, 0 }  ][line width=0.08]  [draw opacity=0] (12,-3) -- (0,0) -- (12,3) -- cycle    ;
				\draw    (258,263) -- (310,263) ;
				\draw [shift={(312,263)}, rotate = 180] [fill={rgb, 255:red, 0; green, 0; blue, 0 }  ][line width=0.08]  [draw opacity=0] (12,-3) -- (0,0) -- (12,3) -- cycle    ;
				\draw   (503.76,155.28) .. controls (510.94,155.15) and (516.86,160.86) .. (516.99,168.04) .. controls (517.12,175.22) and (511.41,181.14) .. (504.23,181.27) .. controls (497.05,181.4) and (491.13,175.69) .. (491,168.51) .. controls (490.87,161.33) and (496.58,155.41) .. (503.76,155.28) -- cycle ;
				\draw   (424.76,156) .. controls (431.94,155.87) and (437.87,161.59) .. (438,168.76) .. controls (438.13,175.94) and (432.41,181.87) .. (425.24,182) .. controls (418.06,182.13) and (412.13,176.41) .. (412,169.24) .. controls (411.87,162.06) and (417.59,156.13) .. (424.76,156) -- cycle ;
				\draw[dashed]   (123.77,159.73) .. controls (174.74,121.19) and (353.2,100.21) .. (423.71,155.17) ;
				\draw [shift={(424.76,156)}, rotate = 218.76] [fill={rgb, 255:red, 0; green, 0; blue, 0 }  ][line width=0.08]  [draw opacity=0] (12,-3) -- (0,0) -- (12,3) -- cycle    ;
				\draw[dashed]   (163.76,159) .. controls (214.74,120.47) and (353.6,100.2) .. (423.71,155.17) ;
				\draw [shift={(424.76,156)}, rotate = 218.76] [fill={rgb, 255:red, 0; green, 0; blue, 0 }  ][line width=0.08]  [draw opacity=0] (12,-3) -- (0,0) -- (12,3) -- cycle    ;
				\draw[dashed]    (203.76,158.28) .. controls (254.73,119.74) and (354,100.2) .. (423.72,155.17) ;
				\draw [shift={(424.76,156)}, rotate = 218.76] [fill={rgb, 255:red, 0; green, 0; blue, 0 }  ][line width=0.08]  [draw opacity=0] (12,-3) -- (0,0) -- (12,3) -- cycle    ;
				\draw[dashed]    (285.24,183.72) .. controls (312.86,221.81) and (377.35,237.84) .. (424.53,182.83) ;
				\draw [shift={(425.24,182)}, rotate = 130.15] [fill={rgb, 255:red, 0; green, 0; blue, 0 }  ][line width=0.08]  [draw opacity=0] (12,-3) -- (0,0) -- (12,3) -- cycle    ;
				\draw[dashed]    (325.24,183) .. controls (352.85,221.08) and (377.75,237.83) .. (424.53,182.83) ;
				\draw [shift={(425.24,182)}, rotate = 130.15] [fill={rgb, 255:red, 0; green, 0; blue, 0 }  ][line width=0.08]  [draw opacity=0] (12,-3) -- (0,0) -- (12,3) -- cycle    ;
				\draw[dashed]    (365.23,182.27) .. controls (374.02,209.2) and (373.07,236.07) .. (424.46,182.81) ;
				\draw [shift={(425.24,182)}, rotate = 133.83] [fill={rgb, 255:red, 0; green, 0; blue, 0 }  ][line width=0.08]  [draw opacity=0] (12,-3) -- (0,0) -- (12,3) -- cycle    ;
				\draw    (438,168.76) -- (489,168.52) ;
				\draw [shift={(491,168.51)}, rotate = 179.73] [fill={rgb, 255:red, 0; green, 0; blue, 0 }  ][line width=0.08]  [draw opacity=0] (12,-3) -- (0,0) -- (12,3) -- cycle    ;
				\draw  (337.97,84.13) .. controls (380,80) and (410.02,100) .. (424.76,156) ;
				\draw [shift={(424.76,156)}, rotate =-110] [fill={rgb, 255:red, 0; green, 0; blue, 0 }  ][line width=0.08]  [draw opacity=0] (12,-3) -- (0,0) -- (12,3) -- cycle    ;
				\draw  (337.97,265.01) .. controls (390,270) and (410.02,230) .. (424.76,184) ;
				\draw [shift={(424.76,184)}, rotate =110] [fill={rgb, 255:red, 0; green, 0; blue, 0 }  ][line width=0.08]  [draw opacity=0] (12,-3) -- (0,0) -- (12,3) -- cycle    ;

				\draw    (503.76,154) .. controls (503,41.57) and (244.6,-53.04) .. (244.98,70.13) ;
				\draw [shift={(245,72)}, rotate = 268.63] [fill={rgb, 255:red, 0; green, 0; blue, 0 }  ][line width=0.08]  [draw opacity=0] (12,-3) -- (0,0) -- (12,3) -- cycle    ;
				\draw    (503.76,154) .. controls (503,41.57) and (323.81,-53.04) .. (324.97,70.13) ;
				\draw [shift={(325,72.01)}, rotate = 268.63] [fill={rgb, 255:red, 0; green, 0; blue, 0 }  ][line width=0.08]  [draw opacity=0] (12,-3) -- (0,0) -- (12,3) -- cycle    ;
				\draw    (503.76,155.28) .. controls (503,42.84) and (163.42,-53.04) .. (162.98,70.12) ;
				\draw [shift={(163,71.99)}, rotate = 268.63] [fill={rgb, 255:red, 0; green, 0; blue, 0 }  ][line width=0.08]  [draw opacity=0] (12,-3) -- (0,0) -- (12,3) -- cycle    ;
				\draw    (504.23,181.27) .. controls (505.99,298.27) and (164,452) .. (163,274.99) ;
				\draw [shift={(163,274.99)}, rotate = 89.68] [fill={rgb, 255:red, 0; green, 0; blue, 0 }  ][line width=0.08]  [draw opacity=0] (12,-3) -- (0,0) -- (12,3) -- cycle    ;
				\draw    (504.23,181.27) .. controls (505.99,298.27) and (246,453.01) .. (245,276) ;
				\draw [shift={(245,276)}, rotate = 89.68] [fill={rgb, 255:red, 0; green, 0; blue, 0 }  ][line width=0.08]  [draw opacity=0] (12,-3) -- (0,0) -- (12,3) -- cycle    ;
				\draw    (504.23,181.27) .. controls (505.99,298.27) and (326,453.01) .. (325,276.01) ;
				\draw [shift={(325,276.01)}, rotate = 89.68] [fill={rgb, 255:red, 0; green, 0; blue, 0 }  ][line width=0.08]  [draw opacity=0] (12,-3) -- (0,0) -- (12,3) -- cycle    ;

				\draw[dashed]    (504.23,156) .. controls (505.99,-100) and (100, -80) .. (65,33) ;
				\draw [shift={(65,33)}, rotate = -65.68] [fill={rgb, 255:red, 0; green, 0; blue, 0 }  ][line width=0.08]  [draw opacity=0] (12,-3) -- (0,0) -- (12,3) -- cycle    ;
				
				\draw[dashed]    (504.23,156) .. controls (505.99,-120) and (-20, -100) .. (54,78) ;
				\draw [shift={(54,78)}, rotate = -110.68] [fill={rgb, 255:red, 0; green, 0; blue, 0 }  ][line width=0.08]  [draw opacity=0] (12,-3) -- (0,0) -- (12,3) -- cycle    ;
				
				\draw[dashed]    (504.23,156) .. controls (505.99,-160) and (-60, -120) .. (54,118) ;
				\draw [shift={(54,118)}, rotate = -113.68] [fill={rgb, 255:red, 0; green, 0; blue, 0 }  ][line width=0.08]  [draw opacity=0] (12,-3) -- (0,0) -- (12,3) -- cycle    ;

				\draw[dashed]    (504.23,184) .. controls (505.99,470.27) and (100, 420) .. (65,315) ;
				\draw [shift={(65,315)}, rotate = 65.68] [fill={rgb, 255:red, 0; green, 0; blue, 0 }  ][line width=0.08]  [draw opacity=0] (12,-3) -- (0,0) -- (12,3) -- cycle    ;
				
				\draw[dashed]    (504.23,184) .. controls (505.99,490.27) and (-20, 450) .. (54,265) ;
				\draw [shift={(54,265)}, rotate = 110.68] [fill={rgb, 255:red, 0; green, 0; blue, 0 }  ][line width=0.08]  [draw opacity=0] (12,-3) -- (0,0) -- (12,3) -- cycle    ;
				
				\draw[dashed]    (504.23,184) .. controls (505.99,500.27) and (-60, 490) .. (54,225) ;
				\draw [shift={(54,225)}, rotate = 110.68] [fill={rgb, 255:red, 0; green, 0; blue, 0 }  ][line width=0.08]  [draw opacity=0] (12,-3) -- (0,0) -- (12,3) -- cycle    ;

				\draw (498,162) node [anchor=north west][inner sep=0.75pt]  [xscale=0.87,yscale=0.87] [align=left] {\tiny{$a$}};
				\draw (420,162) node [anchor=north west][inner sep=0.75pt]  [xscale=0.87,yscale=0.87] [align=left] {\tiny{$p$}};
				\draw (58,33) node [anchor=north west][inner sep=0.75pt]  [xscale=0.87,yscale=0.87]  {\tiny{$c_{1}^{1}$}};
				\draw (58,73) node [anchor=north west][inner sep=0.75pt]  [xscale=0.87,yscale=0.87]  {\tiny{$c_{1}^{2}$}};
				\draw (58,113) node [anchor=north west][inner sep=0.75pt]  [xscale=0.87,yscale=0.87]  {\tiny{$c_{1}^{3}$}};
				\draw (57,210) node [anchor=north west][inner sep=0.75pt]  [xscale=0.87,yscale=0.87]  {\tiny{$c_{2}^{1}$}};
				\draw (57,249) node [anchor=north west][inner sep=0.75pt]  [xscale=0.87,yscale=0.87]  {\tiny{$c_{2}^{2}$}};
				\draw (57,290) node [anchor=north west][inner sep=0.75pt]  [xscale=0.87,yscale=0.87]  {\tiny{$c_{2}^{3}$}};
				\draw (154,73) node [anchor=north west][inner sep=0.75pt]  [xscale=0.87,yscale=0.87]  {\tiny{$x_{1}^{1}$}};
				\draw (236,73) node [anchor=north west][inner sep=0.75pt]  [xscale=0.87,yscale=0.87]  {\tiny{$x_{1}^{2}$}};
				\draw (316,73) node [anchor=north west][inner sep=0.75pt]  [xscale=0.87,yscale=0.87]  {\tiny{$x_{1}^{3}$}};
				\draw (153,250) node [anchor=north west][inner sep=0.75pt]  [xscale=0.87,yscale=0.87]  {\tiny{$x_{2}^{1}$}};
				\draw (236,251) node [anchor=north west][inner sep=0.75pt]  [xscale=0.87,yscale=0.87]  {\tiny{$x_{2}^{2}$}};
				\draw (316,251) node [anchor=north west][inner sep=0.75pt]  [xscale=0.87,yscale=0.87]  {\tiny{$x_{2}^{3}$}};
				\draw (115,161.4) node [anchor=north west][inner sep=0.75pt]  [xscale=0.87,yscale=0.87]  {\tiny{$n_{1}^{1}$}};
				\draw (154,160.4) node [anchor=north west][inner sep=0.75pt]  [xscale=0.87,yscale=0.87]  {\tiny{$n_{1}^{2}$}};
				\draw (194,160.4) node [anchor=north west][inner sep=0.75pt]  [xscale=0.87,yscale=0.87]  {\tiny{$n_{1}^{3}$}};
				\draw (275,158) node [anchor=north west][inner sep=0.75pt]  [xscale=0.87,yscale=0.87]  {\tiny{$n_{2}^{1}$}};
				\draw (315,158) node [anchor=north west][inner sep=0.75pt]  [xscale=0.87,yscale=0.87]  {\tiny{$n_{2}^{2}$}};
				\draw (355,158) node [anchor=north west][inner sep=0.75pt]  [xscale=0.87,yscale=0.87]  {\tiny{$n_{2}^{3}$}};
			\end{tikzpicture}
			\vspace*{-4em}
			\caption{\label{fig:schulze-CCDC}
				WMG corresponding to the election constructed from the $\threesat$ instance from our counterexample with the new reduction.
				Dashed edges have weight two and drawn edges have weight four.}
		\end{center}
		\vspace*{-1em}
	\end{figure}
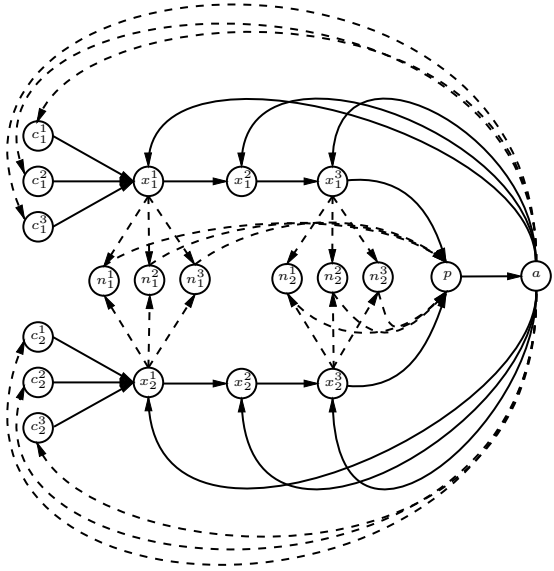
	The graph in Figure~\ref{fig:schulze-CCDC} shows the weighted majority graph for the $\threesat$ instance from our counterexample adapted to the new reduction.
	We claim that $(X,Cl)$ is a yes-instance of $\threesat$ if and only if $((C,V),p,k)$ is a yes-instance of Schulze-\textsc{CCDC} in the nonunique-winner model. 
	
	From left to right, let $(X,Cl)$ be a yes-instance of $\threesat$. 
	Since we have a yes-instance of $\threesat$, we have a truth assignment that makes at least one literal in each clause $Cl_i\in Cl$ true.
	We claim that $p$ can be made a Schulze winner by deleting one literal candidate corresponding to some true literal for each clause.
	The only path with weight four from a clause candidate $c_i^j\in K$ to $p$ is through the literal candidates: from $x_i^1$ via $x_i^2$ to~$x_i^3$.
	Since we deleted one literal candidate for each clause, there no longer exists a weight-$4$ path from a clause candidate to $p$ and $P(p,c) = 2 \geq P(c,p)$ for $c\in K$.
	For each $c\in L \cup \{a\}$, we have $P(p,c) = 4 \ge P(c,p)$.
	Since we deleted only literal candidates where the corresponding literal was assigned to be true, we never have the case that we deleted two literal candidates $x_i^j$, $x_m^n$, which negate each other.
	Thus we still have a weight-$2$ path from $p$ to each negation candidate and $P(p,c) = 2 = P(c,p)$ for each $c\in N$.
	It follows that $((C,V),p,k)$ is a yes-instance of Schulze-\textsc{CCDC} in the nonunique-winner model.
	
	From right to left, let $(X,Cl)$ be a no-instance of $\threesat$.
	Thus, for each assignment of the literals, there exists a clause which is false.
	To ensure that $p$ is a winner of the election, it is necessary that $P(p,c)\geq P(c,p)$ for each $c\in C\setminus\{p\}$.
	Since $P(p,c)= 2 < 4 = P(c,p)$ for each $c\in K$, we have to destroy each path of weight greater than two from the clause candidates to~$p$, in particular the path through the literal candidates $x_i^j \in C_i$, $j \in \{1,2,3\}$.
	Due to the deletion limit $k = \vert Cl \vert$, it is necessary to delete one literal candidate for each clause. 
	Consider any subset of literals of size~$k$ such that for each clause one literal is contained in the set.
	It follows that  this set contains at least two literals, $x_i^j$ and $x_m^n$, that negate each other, otherwise, the formula would be satisfiable and we would have a yes-instance of $\threesat$.
	By deleting the corresponding two literal candidates, we no longer have a path from $p$ to the negation candidates $n\in N_{ijmn}$.
	It follows that $P(p,n) < P(n,p)$ and it is impossible to make $p$ a Schulze winner of the election by deleting at most $k$ candidates.  
	Therefore, $((C,V),p,k)$ is a no-instance of Schulze-\textsc{CCDC} in the nonunique-winner model.
	
	From right to left, let $((C,V), p, k)$ be a yes-instance of Schulze-\textsc{CCDC}. To ensure that $p$ is a winner of the election, it is necessary that $P(p,c)\geq P(c,p)$ for each $c\in C\setminus\{p\}$.
	Since $P(p,c)= 2 < 4 = P(c,p)$ for each $c\in K$, we have to destroy each path of weight greater than two from the clause candidates to~$p$, in particular the path through the literal candidates $x_i^j \in C_i$, $j \in \{1,2,3\}$.
	Due to the deletion limit $k = \vert Cl \vert$, it is necessary to delete one literal candidate for each clause. 
	Consider any subset of literals of size~$k$ such that for each clause one literal is contained in the set.
	If all of these subsets contained at least two literals, $x_i^j$ and $x_m^n$, that negate each other, i.e., if the $\threesat$ formula were not satisfiable, by deleting the corresponding two literal candidates, we would no longer have a path from $p$ to the negation candidates $n\in N_{ijmn}$. That is, it would hold that $P(p,n) < P(n,p)$ and it would be impossible to make $p$ a Schulze winner %
	by deleting at most $k$ candidates. Thus, since we have a yes-instance of Schulze-\textsc{CCDC}, there must exist at least one subset of literals, such that for each clause one literal is contained in the set, i.e., set to true and we have $\threesat$-yes. 
	
	Finally, it is easy to see that Schulze-\textsc{CCDC} is in~\np.
\end{proof}

Previously, Schulze-\textsc{CCDC} was only studied in the nonunique-winner model. By slightly adapting the previous construction (see the supplementary material), we can now also show \np-completeness in the unique-winner model.  

\begin{theorem}\label{thm:schulze-ccdc-unique}
	Schulze-\textsc{CCDC} is~\np-complete in the unique-winner model.
\end{theorem}

\section{Destructive Control by Deleting Candidates and Variants Thereof}\label{sec:variantsdestructivecontrolbydeletingcandidates}

In this section, we examine a peculiarity of standard destructive control by deleting candidates for Schulze elections, which allows us to reduce the number of candidates we need to consider for deletion. %
Using this approach, we are able to re-establish the result that \textsc{DCDC} is polynomial-time solvable for Schulze elections, which Menton and Singh claimed in an early version (v1) of their arXiv preprint~\cite{men-sin:t:manipulation-control-schulze-voting-v1}.
However, Menton and Singh removed this result (and the corresponding result for Schulze-\textsc{DCAC}) from all subsequent versions of the arXiv preprint and from their IJCAI 2013 paper~\cite{men-sin:c:control-complexity-schulze}, stating these two control problems as open.
Additionally, we consider the relationship of variations of destructive control by deleting candidates to variations of $s$-$t$ vertex cut.
All omitted proofs can be found in the supplementary material.

We first introduce a variant of \emph{$s$-$t$ vertex cut} defined by Menton and Singh~\shortcite{men-sin:c:control-complexity-schulze}: \ppvclong.
Recall that for destructive control in Schulze elections to be successful, one candidate must be boosted to beat the despised candidate and that by deleting candidates we can only lower the strength of the strongest path but can never increase it.
Hence, we need to cut paths such that the best remaining path to the despised candidate is stronger than all paths from the despised to our chosen boosted candidate.
Menton and Singh~\shortcite{men-sin:c:control-complexity-schulze} capture this notion of \emph{path-preserving vertex cut}
by defining the problem %
as essentially asking whether there exists an $s$-$t$ vertex cut (destroying the stronger paths from $s$ to~$t$) of size at most~$k$, while at least one path from $t$ to $s$ must remain intact.
When searching for candidates to delete to boost some candidate to beat the despised candidate, it is useful if we can reduce the search space. 
We will at times refer to these candidates as \emph{rivals of the despised candidate}.
For the standard control problem we show that this is indeed possible and the chair can limit the search to the in-neighborhood of the boosted candidate. 

\begin{theorem}\label{thm:ccdc-neighbors}
	If destructive control by deleting candidates is possible for a given Schulze election, then there exists some candidate $c \in C$ who can beat the despised candidate $w$ by only deleting candidates who directly beat $c$ (i.e., are in-neighbors of $c$ in the WMG).\footnote{This result can be extended to any pair of candidates $c,w \in C$ where $P(w,c) \ge P(c,w)$ and
		we have a yes-instance of \ppvclong\ for the graph induced by the paths between these two candidates.}
\end{theorem}

\begin{proof}
	Let $c \in C$ be some candidate where $P(w,c) \ge P(c,w)$, i.e., $w$ has a stronger (or equally strong) path to $c$ than the other way round.
	Assume that $c$ can beat $w$ by deleting the minimal number of candidates needed (with regards to the number necessary for making other candidates beat~$w$).
	We claim that either those deleted candidates are direct neighbors of~$c$, or there exists some other candidate $c^* \in C$ for which we can reach the same goal by deleting equally many or even fewer candidates in the neighborhood of~$c^*$.
	
	First, we define some notation.
	Let $Del^c_w$ be the minimal number of removed candidates needed to make $c$ beat~$w$.
	Note that these candidates form a path-preserving vertex cut.
	We say that $x$ is \textit{before} $z$ if $x$ is closer to $w$ than~$z$.
	Let $N^+(c)$ be the in-neighborhood of a candidate $c$, i.e., $N^+(c)$ contains all candidates with a direct edge to~$c$.
	Finally, we define $Ind^c_w$ as the candidates belonging to the connected component of $c$ as induced by the vertex cut~$Del^c_w$.
	Intuitively, $Ind^c_w$ contains all candidates on stronger paths from $w$ to $c$ that are broken by deleting candidates and where the cut is before the candidate.
	
	Clearly, any $z^* \in Ind^c_w$ also beats $w$ as first $P(z^*,c) \ge P(c,w)$ and thus $P(z^*,w) \ge P(c,w)$, and secondly, no path from $w$ to $z^*$ with strength greater than $P(c,w)$ can exist.
	It follows that $\cardinality{Del^{z^*}_w} \le \cardinality{Del^{c}_w}$.
	We distinguish two cases.
	
	\begin{description}
		\item[Case 1:] \textbf{$Ind^c_w = \emptyset$.}
		Since $Del^c_w$ is a minimal cut, we have $Del^c_w \subseteq N^+(c)$.
		
		\item[Case 2:] \textbf{$Ind^c_w \neq \emptyset$.}		
		Let $F= \{f \in Ind^c_w \mid N^+(f) \cap Del^c_w \neq \emptyset \}$ be the set of all candidates in the connected component of $c$, where we deleted some candidates in the in-neighborhood of~$c$. 
		On the one hand, if $\cardinality{F} = 1$ we have a candidate $f \in F$ who also beats $w$ by deleting $Del^c_w$.
		Since $Del^c_w$ is minimal, it follows that $Del^f_w = Del^c_w$, and therefore, for a successful control action against~$w$, it suffices to delete from $N^+(f)$.
		On the other hand, if $\cardinality{F} > 1$, then $ N^+(f) \cap Del^c_w =  N^+(g) \cap Del^c_w$ for all $f,g \in F$.
		For a contradiction, assume there are two candidates $f,g \in F$ who do not share the same in-neighbors in~$Del^c_w$.
		Deleting either $N^+(f) \cap Del^c_w$ or $N^+(g) \cap Del^c_w$ is sufficient to make the respective candidate beat~$w$.
		Since $N^+(f) \cap Del^c_w \neq N^+(g) \cap Del^c_w$, we have a contradiction to $Del^c_w$ being minimal.\qedhere
	\end{description}
\end{proof}

This result can be used to design an algorithm for Schulze-\textsc{DCDC} in the nonunique-winner model, which runs in polynomial time. Unfortunately, we cannot easily transfer our algorithm to work in the unique-winner model.
The complete algorithm, proof of correctness (based on Theorem~\ref{thm:ccdc-neighbors}) and further explanation as well as an example of why the algorithm cannot solve the unique-winner model can be found in the supplementary material.
Intuitively, for a despised candidate~$d$, the algorithm works by considering every candidate $c\in C\setminus\{d\}$ as a possible rival of $d$ and testing whether $c$ is successful by considering the part of the in-neighborhood of $c$ that intersects the stronger paths from $d$ to $c$ for deletion.

\begin{theorem}\label{thm:schulze-dcdc-our-result}
	In the nonunique-winner model, Schulze-\textsc{DCDC} is solvable in polynomial time.
\end{theorem}

For variants of destructive control, we need to encode the restriction, which is imposed on the deletion set, into the vertex cut, essentially creating corresponding new \ppvclong\ decision problems.
There are several natural restrictions one may apply to Schulze-\textsc{DCDC}, e.g., some candidates may be protected from deletion or labeled candidates must be deleted together.
We give examples and prove the relationship to vertex cut for some of these variants in the supplementary material.

\section{Exact Multimode Control and Control by Replacing}
\label{sec:controlbymultimodeandreplacing}

Now we turn to exact multimode control and control by replacing candidates or voters.
In the former control type, the chair must adhere to alter an \emph{exact} number of candidates or voters or both.
In the latter, the chair must add the \emph{same} number of either candidates or voters as previously have been deleted, i.e., must replace them.
Lorregia et al.~\shortcite{lor-nar-ros-ven-wal:c:replacing-candidates} showed that any voting rule that is resistant to constructive control by deleting candidates and satisfies \emph{insensitivity to bottom-ranked candidates (IBC)}~\cite{lan-mau-pol:c:equilibria-strategic-games} is also resistant to constructive control by replacing candidates. 
A voting rule $\mathcal{E}$ is said to be \emph{insensitive to bottom-ranked candidates} if, given an election $(C,V)$ and a new candidate~$x$, the elections $(C,V)$ and $(C \cup \{x\},V^x)$, where $V^x$ is the list of votes obtained by adding $x$ as the least preferred alternative to each vote in~$V$, have the same winners under $\mathcal{E}$. 
We extend their result to also apply to $\mathcal{E}$-\textsc{ECCAC+DC} and $\mathcal{E}$-\textsc{ECCRC}.

Table~\ref{tab:multimode-replacing-overview} provides an overview of our results.
All (full) proofs omitted in this section can be found in the supplementary material. 

\begin{table}
	\caption[Overview of complexity results for exact multimode control and control by replacing in both winner models.]{
		Overview of complexity results for exact multimode control and control by replacing in both winner models.
		Results marked by $^\spadesuit$ can be found in Corollary~\ref{cor:schulze-eccrc},
		by $^\heartsuit$ in Corollary~\ref{cor:rp-dcrc},
		by $^\diamondsuit$ in Corollary~\ref{cor:schulze-ccrc},
		by $^\star$ in Theorem~\ref{theo:votercontrol} and
		by $^\circ$ in Theorem~\ref{theo:rp-votercontrol}.
	}
	\label{tab:multimode-replacing-overview}
	\begin{tabularx}{\columnwidth}{@{}l@{\hspace*{1mm}}|@{\hspace*{1mm}}C@{\hspace*{1mm}}@{\hspace*{1mm}}C@{\hspace*{1mm}}|@{\hspace*{1mm}}C@{\hspace*{1mm}}@{\hspace*{1mm}}C@{\hspace*{1mm}}}
		\toprule
		& \multicolumn{2}{@{\hspace*{1mm}}c|@{\hspace*{1mm}}}{Schulze} &  \multicolumn{2}{c}{Ranked pairs} \\
		& Const. & Destr. & Constr. & Destr. \\
		\midrule
		Exact\phantom{C}AC+DC & NP-c.$^\spadesuit$   & ? & NP-c.$^\spadesuit$   &  NP-c.$^\heartsuit$ \\
		Exact\phantom{C}RC    & NP-c.$^\spadesuit$   & ? & NP-c.$^\spadesuit$   &  NP-c.$^\heartsuit$ \\
		\phantom{ExactC}RC     & NP-c.$^\diamondsuit$ & ? & NP-c.$^\diamondsuit$ &  NP-c.$^\heartsuit$ \\
		
		Exact\phantom{C}AV+DV & NP-c.$^\star$ &  NP-c.$^\star$ &  NP-c.$^\circ$ &  NP-c.$^\circ$ \\
		\phantom{ExactC}RV & NP-c.$^\star$ &  NP-c.$^\star$ &  NP-c.$^\circ$ &  NP-c.$^\circ$ \\
		\bottomrule		
	\end{tabularx}
\end{table}

\begin{lemma}\label{lem:ibc-plus-ccdc}
	Let $\mathcal{E}$ be a voting rule that satisfies IBC. In both the unique- and nonunique-winner model, if $\mathcal{E}$-\textsc{CCDC} is $\np$-hard, then so are $\mathcal{E}$-\textsc{ECCAC+DC} and $\mathcal{E}$-\textsc{ECCRC}.
\end{lemma}
\begin{proof}
	We reduce $\mathcal{E}$-\textsc{CCDC} to $\mathcal{E}$-\textsc{ECCAC+DC} and $\mathcal{E}$-\textsc{ECCRC}.
	Let $(C,V,p,k)$ be an instance of $\mathcal{E}$-\textsc{CCDC}. 
	Define $C' = C \cup X$ with $X= \{x_1, \ldots, x_k\}$, $D= \{d_1, \ldots, d_{\ell_{AC}}\}$, $\ell_{DC} = \ell_{RC} = k$, and set $\ell_{AC} \in \mathbb{N}$ arbitrarily.
	Let $V' = v\,X\,D$ for every $v \in V$, i.e., add all candidates from $X$ at the bottom of every vote and then add all candidates from $D$ at the bottom of those votes.
	Construct an instance $(C', D, V', p, \ell_{AC}, \ell_{DC})$ of $\mathcal{E}$-\textsc{ECCAC+DC} and an instance $((C', V'), p, \ell_{RC})$ of $\mathcal{E}$-\textsc{ECCRC}.
	
	Assume we have a yes-instance of $\mathcal{E}$-\textsc{CCDC}. Then there exists a set $C^d \subset C$ with $\vert C^d \vert \leq \ell_{DC}$ such that $p$ wins the election $(C \setminus C^d,V)$.
	Delete all candidates in $C^d$ and some candidates in $X^d \subseteq X$ such that exactly $\ell_{DC}$ candidates were deleted and add $\ell_{AC}$ arbitrary candidates $D^a \subseteq D$ to the election.
	Since $\mathcal{E}$ is IBC, candidate $p$ is a winner of the election $(C' \cup D^a \setminus (C^d \cup X^d), V')$ and $\mathcal{E}$-\textsc{ECCAC+DC} is also a yes-instance.
	
	Now assume, we have a no-instance of $\mathcal{E}$-\textsc{CCDC}.
	Then there is no set $C^d \subset C$ with $\vert C^d \vert \leq k$ such that $p$ wins the election $(C \setminus C^d,V)$.
	Since $\mathcal{E}$ is IBC, deleting any set of candidates from $X$ or adding any candidate from $D$ has no influence on the winners of the election.
	Thus, to make $p$ an $\mathcal{E}$ winner, we need to find a set of at most $\ell_{DC}$ candidates in $C$ to delete, which is impossible as $\mathcal{E}$-\textsc{CCDC} is a no-instance.
	
	The argument for $\mathcal{E}$-\textsc{ECCRC} is analogous.
\end{proof}

\begin{lemma}\label{lem:ibc-plus-dcdc}
	Let $\mathcal{E}$ be a voting rule that satisfies IBC. In both the unique- and nonunique-winner model, if $\mathcal{E}$-\textsc{DCDC} is $\np$-hard, then so are $\mathcal{E}$-\textsc{EDCAC+DC} and $\mathcal{E}$-\textsc{EDCRC}.
\end{lemma}
The same construction and proof idea used in the proof of Lemma~\ref{lem:ibc-plus-ccdc} works for Lemma~\ref{lem:ibc-plus-dcdc} as well.

\begin{lemma}\label{lem:schulze-ibc}
	Schulze and ranked pairs are insensitive to bottom-ranked candidates.
\end{lemma}

\begin{corollary}\label{cor:schulze-ccrc}
	In both the unique- and nonunique-winner model, \schulzeproblem{CCRC}  and \rpproblem{CCRC} are \np-complete.
\end{corollary}

\begin{corollary}\label{cor:schulze-eccrc}
	In both the unique- and nonunique-winner model, \schulzeproblem{ECCAC+DC}, \schulzeproblem{ECCRC}, \rpproblem{ECCAC+DC}, and \rpproblem{ECCRC} are \np-complete.
\end{corollary}

\begin{corollary}\label{cor:rp-dcrc}
	In both the unique- and nonunique-winner model, \rpproblem{EDCAC+DC}, \rpproblem{EDCRC}, and \rpproblem{DCRC} are \np-complete.
\end{corollary}

Schulze and ranked pairs are equally resistant to all constructive and destructive control actions of the voter list considered in this paper.

\begin{theorem}\label{theo:votercontrol}
	In both the unique- and nonunique-winner model, \schulzeproblem{ECCAV+DV}, \schulzeproblem{CCRV}, \schulzeproblem{EDCAV+DV}, and \schulzeproblem{DCRV} are \np-complete.
\end{theorem}

\sproofsketch
We reduce from  \textsc{Restricted Exact Cover By 3-Sets} (\rxthreec)~\cite{gon:j:exact-cover-constrained}:
Given a set $B = \{b_1, \ldots, b_{3s}\}$ with $s \geq 1$ and a list $\mathcal{S} = \{S_1, \ldots, S_{3s}\}$, where $S_i = \{b_{i,1}, b_{i,2}, b_{i,3}\}$ and $S_i \subseteq B$ for all $S_i \in \mathcal{S}$ and each $b_j$ is contained in exactly three sets $S_i \in \mathcal{S}$, does there exist an exact cover, i.e., a sublist $\mathcal{S}' \subseteq \mathcal{S}$ such that each $b_i \in B$ occurs in exactly one $S_i \in \mathcal{S}'$?

Let $\lAV = \lDV = s$ for the \schulzeproblem{ECCAV+DV} and \schulzeproblem{EDCAV+DV} instances we construct, and let $\lRV = s$ for the \schulzeproblem{CCRV} and \schulzeproblem{DCRV} instances.
Further, let $\lConst \gg s$ be a constant much greater than~$s$.\footnote{The precise value of $\lConst$ is not important; all that matters is that when used as an edge weight in a WMG, $\lConst$ is large enough, such that any edge changed by the control actions must not change in direction, i.e., the sign of the edge weight must not flip.
	Recall that the strength of a path in a WMG is specified as the weight of the weakest edge on the path.}
From $(B,\mathcal{S})$ we construct an election $(C,V)$ as depicted in Figure~\ref{fig:ccrv-election-construction}.
Note that candidate $w$ is the unique winner of the election.
	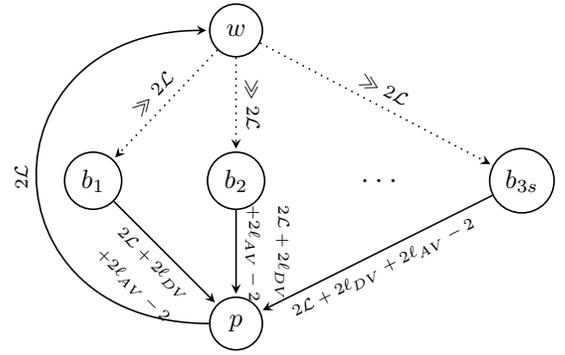
\begin{figure}
		\centering
		\begin{tikzpicture}[node distance=1.9cm, semithick]
			\tikzstyle{voteedge}=[->,>=stealth,shorten >=1pt]
			\tikzstyle{gt2ledge}=[->,>=stealth,shorten >=1pt, dotted]
			\tikzstyle{candidate}=[minimum size=0.7cm, draw, circle]
			\tikzstyle{dots}=[draw=none]
			
			\node[candidate](b2) 							{$b_2$};
			\node[candidate](b1) 		[left of=b2] 		{$b_1$};
			\node[dots](bdots) 			[right of=b2] 		{\Large$\ldots$};
			\node[candidate](b3s) 		[right of=bdots] 	{$b_{3s}$};
			\node[candidate](p) 		[below of=b2] 		{$p$};
			
			\node[candidate, node distance=2cm](w)         [above of=b2] {$w$};
			
			\draw[gt2ledge] (w) -- node[midway, above, sloped] {\scriptsize{$\gg 2\lConst$}} (b1);
			\draw[gt2ledge] (w) -- node[midway, above, sloped] {\scriptsize{$\gg 2\lConst$}} (b2);
			\draw[gt2ledge] (w) -- node[midway, above, sloped] {\scriptsize{$\gg 2\lConst$}} (b3s);
			
			\draw[voteedge,looseness=2] (p) to [bend left=90] node[midway, above, sloped] {\scriptsize{$2\lConst$}} (w);			
			
			\draw[voteedge] (b1) -- node[align=left, sloped, midway, below] {{\tiny $2\lConst+2\lDV$}\\ {\tiny $+2\lAV-2$}} (p);
			\draw[voteedge] (b2) -- node[align=left, sloped, midway,above] {{\tiny $2\lConst+2\lDV$}\\ {\tiny $+2\lAV-2$}}(p);
			\draw[voteedge] (b3s) -- node[sloped, midway, below] {{\tiny $2\lConst+2\lDV+2\lAV-2$}} (p);			
		\end{tikzpicture}
		\vspace*{-0.5em}
		\caption{The WMG of the election from the proof of Theorem~\ref{theo:votercontrol}.}
		\label{fig:ccrv-election-construction}
	\end{figure}
	The list of additional votes $U$ contains one vote
	\[
	S_i\, p\, (B \setminus S_i)\, w
	\hspace{0.5cm} \text{for each}\ S_i \in \mathcal{S}.
	\]
	Let $p$ be the distinguished candidate for the constructive case and $w$ be the despised candidate for the destructive case.
	\eproofsketch
	
	By adapting the above construction, %
	we obtain the same results for ranked pairs. %
	
	\begin{theorem}\label{theo:rp-votercontrol}
		In both the unique- and nonunique-winner model, \rpproblem{ECCAV+DV}, \rpproblem{CCRV}, \rpproblem{EDCAV+DV}, and \rpproblem{DCRV} are \np-complete.
	\end{theorem}

	\section{Conclusion and Future Work}\label{sec:conclusions}
	
	We have studied %
	electoral control for Schulze and ranked pairs elections:
	After fixing a flaw in the proof of Menton and Singh~\shortcite[Theorem~2.2]{men-sin:c:control-complexity-schulze} for \schulzeproblem{CCDC} and extending the new construction to the unique-winner model (Section~\ref{sec:control-ccdc}), 
	we study variants of $s$-$t$ vertex cuts in graphs that are related to destructive control by deleting candidates in Schulze elections, and re-establish a vulnerability result on the corresponding problem (Section~\ref{sec:variantsdestructivecontrolbydeletingcandidates}). 
	Finally, in Section~\ref{sec:controlbymultimodeandreplacing}, we established a number of resistance results for control by replacing candidates or voters and exact multimode control. %
	Hence, our work establishes both polynomial-time algorithms and \np-completeness results.
	Tables~\ref{tab:controlresults} and~\ref{tab:multimode-replacing-overview} provide a summary of our results in the context of related known results.
		
	However, multiple variants of destructive control of the candidate set (see Section~\ref{sec:variantsdestructivecontrolbydeletingcandidates}),
	such as destructive control by adding candidates both in the unique- and nonunique-winner model and by deleting candidates in the unique-winner model, remain open for Schulze elections.
	Lastly, to the best of our knowledge, most cases of control by partition %
	are yet to be solved for ranked pairs elections.

\section*{Ethical Statement}

The subject of this paper---electoral control---is an ethically sensitive topic.
The goal of our work is not to help strategic parties to control elections but to `audit' some voting rules with respect to their vulnerability to electoral control.
In particular, our resistance results are clearly helpful to society since they may help protect society from attacks against elections, and our vulnerability results are clearly helpful to society since they can be used to
avoid voting rules that are vulnerable to certain types of electoral control that realistically may occur in the context of an election.
 Hence, the purpose of our research is to provide detailed information on the societal impact of voting rules and to ensure transparency.

\section*{Acknowledgments}
We are grateful to Ulle Endriss, Piotr Faliszewski, Edith Hemaspaandra, Lane A. Hemaspaandra, and J{\'e}r\^{o}me Lang for their ethical advice.
We thank the anonymous reviewers for helpful comments.
This work was supported in part by DFG research grant RO\nobreakdash-1202/21\nobreakdash-1. 
\bibliographystyle{named}
\bibliography{control}

\clearpage

\appendix

\section{Deferred Example from Section~\ref{sec:introduction}}

\begin{example}\label{ex:exact-vs-non-exact}
	The following simple example illustrates the difference between the exact and non exact versions of electoral control.
	Given an election $(\{p,d\}, V)$ and unregistered candidates $D=\{a,b\}$. The pairwise comparisons are illustrated in Figure~\ref{fig:necessity-del-recently-added}. Clearly, $d$ wins in the original election.
	Consider constructive control by adding candidates where $k = 2$ and candidate $p$ is the preferred winner. For non exact control a chair may add $a$, which $p$ beats directly and this also introduces a stronger path from $p$ to $d$, making $p$ the unique winner of the election. However, should the chair also add $b$, the successful control action is immediately reversed, as $b$ beats $p$ in the pairwise contest. Thus, exact control is not possible in this setting.

	\begin{figure}[ht]
		\centering
		\begin{tikzpicture}[node distance={20mm}, 
			main/.style = {draw, circle, scale=0.8},
			el/.style = {inner sep=2pt, align=left, sloped},
			every label/.append style = {font=\normalsize}]
			\tikzset{edge/.style = {->,> = latex'}}
			
			\node[main, minimum size=0.7cm] (d) {$d$};
			\node[main, minimum size=0.7cm] [above left of=d] (a) {$a$};
			\node[main, minimum size=0.7cm] [below left of=d] (b) {$b$};
			\node[main, minimum size=0.7cm] [below left of=a] (p) {$p$};
			
			\draw[edge] (d) -- (p) node[midway, above] {$2$};
			\draw[edge, dashed] (a) -- (d) node[midway, right] {$4$};
			\draw[edge, dashed] (p) -- (a) node[midway, left] {$4$};
			\draw[edge, dashed] (b) -- (p) node[midway, left] {$6$};
			\draw[edge, dashed] (d) -- (b) node[midway, right] {$6$};
			
		\end{tikzpicture}
		\caption{A WMG of the election in Example~\ref{ex:exact-vs-non-exact}. Candidates from the set of unregistered candidates $D$ are depicted using dashed lines.}
		\label{fig:example-exact-vs-non-exact}
	\end{figure}
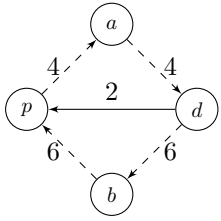
\end{example}

\section{Deferred Example from Section~\ref{sec:prelim}}\label{sec:apdx-prelim}

\begin{example}\label{ex:schulzemethod}
	Consider an election $(C,V)$ with $C= \{a,b,c,d\}$ and the following votes:
	\begin{align*}
		&4\times v_1:\ a\ c\ b\ d,\\
		&2\times v_2:\ d\ a\ c\ b,\\
		&3\times v_3:\ d\ c\ a\ b,\\
		&2\times v_4:\ b\ d\ a\ c.
	\end{align*}
	First, we determine $\numberrankedabove{V}{x}{y}$ for each pair of candidates $x,y \in C$ to build the WMG:
	\begin{align*}
		&\numberrankedabove{V}{a}{b} = 9, &\numberrankedabove{V}{a}{c} = 8, \hspace*{2em} &\numberrankedabove{V}{a}{d} = 4, \\
		& \numberrankedabove{V}{b}{a} = 2, & \numberrankedabove{V}{b}{c} = 2, \hspace*{2em} & \numberrankedabove{V}{b}{d} = 6, \\
		&\numberrankedabove{V}{c}{a} = 3, &\numberrankedabove{V}{c}{b} = 9, \hspace*{2em} &\numberrankedabove{V}{c}{d} = 4,\\
		& \numberrankedabove{V}{d}{a} = 7, & \numberrankedabove{V}{d}{b} = 5, \hspace*{2em} & \numberrankedabove{V}{d}{c} = 7.
	\end{align*}
	The WMG is shown in Figure~\ref{fig:wmg}. The strengths of the strongest paths for each pair of candidates $x,y \in C$ are given in Table~\ref{tab:stropa}. 
	Since $P(d,x) > P(x,d)$ for each $x \in C\setminus \{d\}$, candidate $d$ is the unique Schulze winner of the election.
	\begin{figure}[th]
		\begin{minipage}[h!]{0.4\columnwidth}
			\centering
			\begin{tikzpicture}[node distance={20mm}, 
				main/.style = {draw, circle, scale=0.8},
				el/.style = {inner sep=2pt, align=left, sloped},
				every label/.append style = {font=\normalsize}]
				\tikzset{edge/.style = {->,> = latex'}}
				
				\node[main,  minimum size=0.7cm] (a) {a};
				\node[main,  minimum size=0.7cm] [right of=a] (b) {b};
				\node[main,  minimum size=0.7cm] [below of=a] (c) {c};
				\node[main,  minimum size=0.7cm] [right of=c] (d) {d};
				
				\draw[edge] (a) -- (b) node[midway, above] {$7$};
				\draw[edge] (a) -- (c) node[midway, left] {$5$};
				\draw[edge] (b) -- (d) node[midway, right] {$1$};
				\draw[edge] (c) -- (b) node[pos=0.75, below] {$7$};
				\draw[edge] (d) -- (a) node[pos=0.75, below] {$3$};
				\draw[edge] (d) -- (c) node[midway, below] {$3$};
			\end{tikzpicture}
			\captionof{figure}{Weighted majority graph in Example~\ref{ex:schulzemethod}.}
			\label{fig:wmg}
		\end{minipage}
		\hspace{\fill}
		\begin{minipage}[h!]{0.55\columnwidth}
			\centering
			\begin{tabularx}{\textwidth}{c|CCCC}
				$\strongestpath{x}{y}$ & $a$ & $b$ & $c$ & $d$\\
				\midrule			
				$a$  & -- 			& $\mathbf{7}$  & $\mathbf{5}$  & $1$ \\
				$b$  & $1$ 			& -- 			& $1$ 			& $1$ \\		
				$c$  & $1$  		& $\mathbf{7}$ 	& -- 			& $1$ \\
				$d$  & $\mathbf{3}$ & $\mathbf{3}$ 	& $\mathbf{3}$	& --  \\
			\end{tabularx}
			\captionof{table}{Strengths of the strongest paths in Example~\ref{ex:schulzemethod}. Bold font indicates $P(x,y) \ge P(y,x)$.}
			\label{tab:stropa}
		\end{minipage}
	\end{figure}
\end{example}

\begin{example}\label{ex:rankedpairsmethod}
	Consider the election described in Example~\ref{ex:schulzemethod}.
	By using lexicographical, i.e., here alphabetical, tie-breaking, we get the weight order shown in Table~\ref{tab:weightorder}.
	We now start by considering the first pair $(a,b)$.
	The corresponding edge $(a,b)$ can of course be inserted into the directed graph~$G$.
	In the next four steps, one after another we consider the pairs (or, edges) $(c,b)$, $(a,c)$, $(d,a)$, and $(d,c)$, which can all be added to $G$ since none of them creates a cycle.
	Only the last edge $(b,d)$ cannot be inserted since we would then get a cycle among $b$, $c$, and~$d$.
	The directed graph $G$ is depicted in Figure~\ref{fig:directedgraphrankedpairs}.
	Being the source of the graph, $d$ is the ranked pairs winner of the election.
	
	\begin{figure}[th]
		\begin{minipage}[h!]{0.55\columnwidth}
			\centering
			\begin{tabularx}{\textwidth}{c|c|c}
				& pair $(c_i, c_j)$ & $\differencerankedabove{V}{c_i}{c_j}$ \\
				\midrule
				1. & $(a,b)$ & $7$ \\
				2. & $(c,b)$ & $7$ \\
				3. & $(a,c)$ & $5$ \\
				4. & $(d,a)$ & $3$ \\
				5. & $(d,c)$ & $3$ \\
				6. & $(b,d)$ & $1$ \\
			\end{tabularx}
			\captionof{table}{Candidate pairs ordered by difference in pairwise comparison in Example~\ref{ex:rankedpairsmethod}.}
			\label{tab:weightorder}
		\end{minipage}
		\hspace{\fill}
		\begin{minipage}[h!]{0.4\columnwidth}
			\centering
			\begin{tikzpicture}[node distance={20mm}, 
				main/.style = {draw, circle, scale=0.8},
				el/.style = {inner sep=2pt, align=left, sloped},
				every label/.append style = {font=\normalsize}]
				\tikzset{edge/.style = {->,> = latex'}}
				
				\node[main,  minimum size=0.7cm] (a) {a};
				\node[main,  minimum size=0.7cm] [right of=a] (b) {b};
				\node[main,  minimum size=0.7cm] [below of=a] (c) {c};
				\node[main,  minimum size=0.7cm] [right of=c] (d) {d};
				
				\draw[edge] (a) -- (b) node[midway, above] {};
				\draw[edge] (a) -- (c) node[midway, left] {};
				\draw[edge] (c) -- (b) node[pos=0.75, below] {};
				\draw[edge] (d) -- (a) node[pos=0.75, below] {};
				\draw[edge] (d) -- (c) node[midway, below] {};
			\end{tikzpicture}
			\captionof{figure}{Directed graph $G$ in Example~\ref{ex:rankedpairsmethod}.}
			\label{fig:directedgraphrankedpairs}
		\end{minipage}
	\end{figure}
\end{example}

\begin{example}\label{ex:deleterecentlyadded}
	Consider a Schulze election where apart from the preferred candidate $p$ and the current winner $w$ only two more candidates, $c_1$ and~$c_2$, are present.
	Let the election have one stronger path from $w$ to $p$ through $c_1$ and one weaker path from $p$ to $w$ through~$c_2$.
	The set $D$ of as yet unregistered candidates to be added contains only~$d$, who is on a path from $w$ to $p$ that is as strong as the path from $w$ to $p$ through~$c_1$.
	This scenario is illustrated in Figure~\ref{fig:necessity-del-recently-added}.
	Let $\ell_{DC} = 2$ and $\ell_{AC} = 1$.
	If we do not allow deletion of recently added candidates, we are forced to delete both candidates $c_1$ and $c_2$ in any exact control action and to add the one candidate from~$D$.
	But this would make it impossible for $p$ to win.
	However, if we are allowed to delete $d$ after adding (thus in effect reducing both $\ell_{DC}$ and $\ell_{AC}$ by one), our control action---first adding $d$ and then deleting $c_1$ and~$d$---is successful.

	\begin{figure}[ht]
		\centering
		\begin{tikzpicture}[node distance={20mm}, 
			main/.style = {draw, circle, scale=0.8},
			el/.style = {inner sep=2pt, align=left, sloped},
			every label/.append style = {font=\normalsize}]
			\tikzset{edge/.style = {->,> = latex'}}
			
			\node[main, minimum size=0.7cm] (w) {$w$};
			\node[main, minimum size=0.7cm] [below left of=w] (c1) {$c_1$};
			\node[main, dashed, minimum size=0.7cm] [left of=c1] (d1) {$d$};
			\node[main, minimum size=0.7cm] [below right of=c1] (p) {$p$};
			\node[main, minimum size=0.7cm] [below right of=w] (c2) {$c_2$};
			
			\draw[edge] (w) -- (c1) node[midway, right] {$4$};
			\draw[edge, dashed] (w) -- (d1) node[midway, left] {$4$};
			\draw[edge] (c1) -- (p) node[midway, right] {$4$};
			\draw[edge, dashed] (d1) -- (p) node[midway, left] {$4$};
			\draw[edge] (p) -- (c2) node[midway, right] {$2$};
			\draw[edge] (c2) -- (w) node[midway, right] {$2$};
		\end{tikzpicture}
		\caption{A WMG of the election in Example~\ref{ex:deleterecentlyadded}, illustrating that it may be necessary to delete recently added candidates for an exact control action to be successful.}
		\label{fig:necessity-del-recently-added}
	\end{figure}
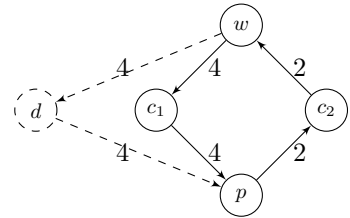
\end{example}

\section{Deferred Counterexample and Proof of Section~\ref{sec:control-ccdc}}\label{sec:appendix_ccdc}

\begin{example}\label{ex:ccdc}
	We start with a yes-instance of $\threesat$, which will be mapped to a 
	no-instance of Schulze-\textsc{CCDC} by their reduction.  
	Let $(X,Cl)$ be our given $\threesat$ instance, with $X=\{x_1,x_2,x_3\}$ and $Cl=\{(x_1\vee x_2 \vee \neg {x}_3), \left(\neg {x}_{1} \vee x_{2} \vee x_3\right)\}$, i.e., we consider the CNF formula
	\[
	\varphi = \left( x_{1} \vee x_{2} \vee \neg {x}_{3}\right) \wedge
	\left(\neg {x}_{1} \vee x_{2} \vee x_3\right).
	\]
	According to the proof of \cite[Thm.~2.2]{men-sin:c:control-complexity-schulze}, we construct an instance $((C, V'), p, \ell)$ of Schulze-\textsc{CCDC} as follows.
	The deletion limit is $k = 2$, the number of clauses~$\vert Cl \vert$.
	The set of candidates, $C = K \cup L \cup N \cup \{p,a\}$, consists of
	\begin{itemize}
		\item the clause candidates $c_1^1$, $c_1^2$, $c_1^3,c_2^1$, $c_2^2$,  and $c_2^3$,
		\item the literal candidates $x_1^1$, $x_1^2$, $x_1^3$, $x_2^1$, $x_2^2$, and $x_2^3$,
		\item the negation candidates $n_{1,1,2,1}^{1}$, $n_{1,1,2,1}^{2}$, $n_{1,1,2,1}^{3}$, $n_{1,3,2,3}^{1}$,\linebreak[4] $\ n_{1,3,2,3}^{2}$, and $n_{1,3,2,3}^{3}$ (abbreviate by $n_{1}^{1}$, $n_{1}^{2}$, $n_{1}^{3}$, $n_{2}^{1}$, $n_{2}^{2}$, and~$n_{2}^{3}$),
		\item the distinguished candidate $p$ and the additional candidate~$a$.
	\end{itemize}
	The preferences are represented by the WMG in Figure~\ref{fig:congra} using McGarvey's trick~\cite{mcg:j:election-graph}.
	Each edge in this graph has a weight of two.
	
	\begin{figure}[h!]
		\centering
		\tikzset{every picture/.style={line width=0.75pt}} 
		\vspace*{-0.8cm}
		\small\begin{tikzpicture}[x=0.43pt,y=0.43pt,yscale=-1,xscale=1]
			\draw   (53,45) .. controls (53,37.82) and (58.82,32) .. (66,32) .. controls (73.18,32) and (79,37.82) .. (79,45) .. controls (79,52.18) and (73.18,58) .. (66,58) .. controls (58.82,58) and (53,52.18) .. (53,45) -- cycle ;
			\draw   (53,85) .. controls (53,77.82) and (58.82,72) .. (66,72) .. controls (73.18,72) and (79,77.82) .. (79,85) .. controls (79,92.18) and (73.18,98) .. (66,98) .. controls (58.82,98) and (53,92.18) .. (53,85) -- cycle ;
			\draw   (53,125) .. controls (53,117.82) and (58.82,112) .. (66,112) .. controls (73.18,112) and (79,117.82) .. (79,125) .. controls (79,132.18) and (73.18,138) .. (66,138) .. controls (58.82,138) and (53,132.18) .. (53,125) -- cycle ;
			\draw   (53,222) .. controls (53,214.82) and (58.82,209) .. (66,209) .. controls (73.18,209) and (79,214.82) .. (79,222) .. controls (79,229.18) and (73.18,235) .. (66,235) .. controls (58.82,235) and (53,229.18) .. (53,222) -- cycle ;
			\draw   (53,262) .. controls (53,254.82) and (58.82,249) .. (66,249) .. controls (73.18,249) and (79,254.82) .. (79,262) .. controls (79,269.18) and (73.18,275) .. (66,275) .. controls (58.82,275) and (53,269.18) .. (53,262) -- cycle ;
			\draw   (53,302) .. controls (53,294.82) and (58.82,289) .. (66,289) .. controls (73.18,289) and (79,294.82) .. (79,302) .. controls (79,309.18) and (73.18,315) .. (66,315) .. controls (58.82,315) and (53,309.18) .. (53,302) -- cycle ;
			\draw   (325,72.01) .. controls (332.18,72.01) and (338,77.83) .. (338,85.01) .. controls (338,92.19) and (332.18,98.01) .. (325,98.01) .. controls (317.82,98.01) and (312,92.18) .. (312,85) .. controls (312,77.83) and (317.82,72.01) .. (325,72.01) -- cycle ;
			\draw   (245,72) .. controls (252.18,72) and (258,77.82) .. (258,85) .. controls (258,92.18) and (252.18,98) .. (245,98) .. controls (237.82,98) and (232,92.18) .. (232,85) .. controls (232,77.82) and (237.82,72) .. (245,72) -- cycle ;
			\draw   (163,71.99) .. controls (170.18,71.99) and (176,77.82) .. (176,85) .. controls (176,92.17) and (170.18,97.99) .. (163,97.99) .. controls (155.82,97.99) and (150,92.17) .. (150,84.99) .. controls (150,77.81) and (155.82,71.99) .. (163,71.99) -- cycle ;
			\draw   (325,250.01) .. controls (332.18,250.01) and (338,255.83) .. (338,263.01) .. controls (338,270.19) and (332.18,276.01) .. (325,276.01) .. controls (317.82,276.01) and (312,270.18) .. (312,263) .. controls (312,255.83) and (317.82,250.01) .. (325,250.01) -- cycle ;
			\draw   (245,250) .. controls (252.18,250) and (258,255.82) .. (258,263) .. controls (258,270.18) and (252.18,276) .. (245,276) .. controls (237.82,276) and (232,270.18) .. (232,263) .. controls (232,255.82) and (237.82,250) .. (245,250) -- cycle ;
			\draw   (163,248.99) .. controls (170.18,248.99) and (176,254.82) .. (176,262) .. controls (176,269.17) and (170.18,274.99) .. (163,274.99) .. controls (155.82,274.99) and (150,269.17) .. (150,261.99) .. controls (150,254.81) and (155.82,248.99) .. (163,248.99) -- cycle ;
			\draw    (79,45) -- (148.26,84.01) ;
			\draw [shift={(150,84.99)}, rotate = 209.39] [fill={rgb, 255:red, 0; green, 0; blue, 0 }  ][line width=0.08]  [draw opacity=0] (12,-3) -- (0,0) -- (12,3) -- cycle    ;
			\draw    (79,85) -- (148,84.99) ;
			\draw [shift={(150,84.99)}, rotate = 179.99] [fill={rgb, 255:red, 0; green, 0; blue, 0 }  ][line width=0.08]  [draw opacity=0] (12,-3) -- (0,0) -- (12,3) -- cycle    ;
			\draw    (79,125) -- (148.26,85.97) ;
			\draw [shift={(150,84.99)}, rotate = 150.6] [fill={rgb, 255:red, 0; green, 0; blue, 0 }  ][line width=0.08]  [draw opacity=0] (12,-3) -- (0,0) -- (12,3) -- cycle    ;
			\draw    (79,222) -- (148.26,261.01) ;
			\draw [shift={(150,261.99)}, rotate = 209.39] [fill={rgb, 255:red, 0; green, 0; blue, 0 }  ][line width=0.08]  [draw opacity=0] (12,-3) -- (0,0) -- (12,3) -- cycle    ;
			\draw    (79,262) -- (148,261.99) ;
			\draw [shift={(150,261.99)}, rotate = 179.99] [fill={rgb, 255:red, 0; green, 0; blue, 0 }  ][line width=0.08]  [draw opacity=0] (12,-3) -- (0,0) -- (12,3) -- cycle    ;
			\draw    (80,301) -- (148.25,262.96) ;
			\draw [shift={(150,261.99)}, rotate = 150.87] [fill={rgb, 255:red, 0; green, 0; blue, 0 }  ][line width=0.08]  [draw opacity=0] (12,-3) -- (0,0) -- (12,3) -- cycle    ;
			\draw   (203.76,158.28) .. controls (210.94,158.15) and (216.86,163.86) .. (216.99,171.04) .. controls (217.12,178.22) and (211.41,184.14) .. (204.23,184.27) .. controls (197.05,184.4) and (191.13,178.69) .. (191,171.51) .. controls (190.87,164.33) and (196.58,158.41) .. (203.76,158.28) -- cycle ;
			\draw   (163.76,159) .. controls (170.94,158.87) and (176.87,164.59) .. (177,171.76) .. controls (177.13,178.94) and (171.41,184.87) .. (164.24,185) .. controls (157.06,185.13) and (151.13,179.41) .. (151,172.24) .. controls (150.87,165.06) and (156.59,159.13) .. (163.76,159) -- cycle ;
			\draw   (123.77,159.73) .. controls (130.95,159.6) and (136.87,165.31) .. (137,172.49) .. controls (137.13,179.67) and (131.42,185.59) .. (124.24,185.72) .. controls (117.06,185.85) and (111.14,180.14) .. (111.01,172.96) .. controls (110.88,165.78) and (116.59,159.86) .. (123.77,159.73) -- cycle ;
			\draw   (364.76,156.28) .. controls (371.94,156.15) and (377.86,161.86) .. (377.99,169.04) .. controls (378.12,176.22) and (372.41,182.14) .. (365.23,182.27) .. controls (358.05,182.4) and (352.13,176.69) .. (352,169.51) .. controls (351.87,162.33) and (357.58,156.41) .. (364.76,156.28) -- cycle ;
			\draw   (324.76,157) .. controls (331.94,156.87) and (337.87,162.59) .. (338,169.76) .. controls (338.13,176.94) and (332.41,182.87) .. (325.24,183) .. controls (318.06,183.13) and (312.13,177.41) .. (312,170.24) .. controls (311.87,163.06) and (317.59,157.13) .. (324.76,157) -- cycle ;
			\draw   (284.77,157.73) .. controls (291.95,157.6) and (297.87,163.31) .. (298,170.49) .. controls (298.13,177.67) and (292.42,183.59) .. (285.24,183.72) .. controls (278.06,183.85) and (272.14,178.14) .. (272.01,170.96) .. controls (271.88,163.78) and (277.59,157.86) .. (284.77,157.73) -- cycle ;
			\draw    (163,97.99) -- (124.84,158.04) ;
			\draw [shift={(123.77,159.73)}, rotate = 302.43] [fill={rgb, 255:red, 0; green, 0; blue, 0 }  ][line width=0.08]  [draw opacity=0] (12,-3) -- (0,0) -- (12,3) -- cycle    ;
			\draw    (163,97.99) -- (163.74,157) ;
			\draw [shift={(163.76,159)}, rotate = 269.28] [fill={rgb, 255:red, 0; green, 0; blue, 0 }  ][line width=0.08]  [draw opacity=0] (12,-3) -- (0,0) -- (12,3) -- cycle    ;
			\draw    (163,97.99) -- (202.64,156.62) ;
			\draw [shift={(203.76,158.28)}, rotate = 235.94] [fill={rgb, 255:red, 0; green, 0; blue, 0 }  ][line width=0.08]  [draw opacity=0] (12,-3) -- (0,0) -- (12,3) -- cycle    ;
			\draw    (163,248.99) -- (125.29,187.43) ;
			\draw [shift={(124.24,185.72)}, rotate = 58.51] [fill={rgb, 255:red, 0; green, 0; blue, 0 }  ][line width=0.08]  [draw opacity=0] (12,-3) -- (0,0) -- (12,3) -- cycle    ;
			\draw    (163,248.99) -- (164.2,187) ;
			\draw [shift={(164.24,185)}, rotate = 91.1] [fill={rgb, 255:red, 0; green, 0; blue, 0 }  ][line width=0.08]  [draw opacity=0] (12,-3) -- (0,0) -- (12,3) -- cycle    ;
			\draw    (163,248.99) -- (203.15,185.96) ;
			\draw [shift={(204.23,184.27)}, rotate = 122.5] [fill={rgb, 255:red, 0; green, 0; blue, 0 }  ][line width=0.08]  [draw opacity=0] (12,-3) -- (0,0) -- (12,3) -- cycle    ;
			\draw    (325,97.99) -- (285.89,156.07) ;
			\draw [shift={(284.77,157.73)}, rotate = 303.96] [fill={rgb, 255:red, 0; green, 0; blue, 0 }  ][line width=0.08]  [draw opacity=0] (12,-3) -- (0,0) -- (12,3) -- cycle    ;
			\draw    (325,97.99) -- (324.77,155) ;
			\draw [shift={(324.76,157)}, rotate = 270.23] [fill={rgb, 255:red, 0; green, 0; blue, 0 }  ][line width=0.08]  [draw opacity=0] (12,-3) -- (0,0) -- (12,3) -- cycle    ;
			\draw    (325,97.99) -- (363.63,154.62) ;
			\draw [shift={(364.76,156.28)}, rotate = 235.7] [fill={rgb, 255:red, 0; green, 0; blue, 0 }  ][line width=0.08]  [draw opacity=0] (12,-3) -- (0,0) -- (12,3) -- cycle    ;
			\draw    (326,249.99) -- (286.29,185.43) ;
			\draw [shift={(285.24,183.72)}, rotate = 58.41] [fill={rgb, 255:red, 0; green, 0; blue, 0 }  ][line width=0.08]  [draw opacity=0] (12,-3) -- (0,0) -- (12,3) -- cycle    ;
			\draw    (326,249.99) -- (325.26,185) ;
			\draw [shift={(325.24,183)}, rotate = 89.34] [fill={rgb, 255:red, 0; green, 0; blue, 0 }  ][line width=0.08]  [draw opacity=0] (12,-3) -- (0,0) -- (12,3) -- cycle    ;
			\draw    (325,250.01) -- (364.21,183.99) ;
			\draw [shift={(365.23,182.27)}, rotate = 120.71] [fill={rgb, 255:red, 0; green, 0; blue, 0 }  ][line width=0.08]  [draw opacity=0] (12,-3) -- (0,0) -- (12,3) -- cycle    ;
			\draw    (176,85) -- (230,85) ;
			\draw [shift={(232,85)}, rotate = 180] [fill={rgb, 255:red, 0; green, 0; blue, 0 }  ][line width=0.08]  [draw opacity=0] (12,-3) -- (0,0) -- (12,3) -- cycle    ;
			\draw    (258,85) -- (310,85) ;
			\draw [shift={(312,85)}, rotate = 180] [fill={rgb, 255:red, 0; green, 0; blue, 0 }  ][line width=0.08]  [draw opacity=0] (12,-3) -- (0,0) -- (12,3) -- cycle    ;
			\draw    (176,263) -- (230,263) ;
			\draw [shift={(232,263)}, rotate = 180] [fill={rgb, 255:red, 0; green, 0; blue, 0 }  ][line width=0.08]  [draw opacity=0] (12,-3) -- (0,0) -- (12,3) -- cycle    ;
			\draw    (258,263) -- (310,263) ;
			\draw [shift={(312,263)}, rotate = 180] [fill={rgb, 255:red, 0; green, 0; blue, 0 }  ][line width=0.08]  [draw opacity=0] (12,-3) -- (0,0) -- (12,3) -- cycle    ;
			\draw   (503.76,155.28) .. controls (510.94,155.15) and (516.86,160.86) .. (516.99,168.04) .. controls (517.12,175.22) and (511.41,181.14) .. (504.23,181.27) .. controls (497.05,181.4) and (491.13,175.69) .. (491,168.51) .. controls (490.87,161.33) and (496.58,155.41) .. (503.76,155.28) -- cycle ;
			\draw   (424.76,156) .. controls (431.94,155.87) and (437.87,161.59) .. (438,168.76) .. controls (438.13,175.94) and (432.41,181.87) .. (425.24,182) .. controls (418.06,182.13) and (412.13,176.41) .. (412,169.24) .. controls (411.87,162.06) and (417.59,156.13) .. (424.76,156) -- cycle ;
			\draw   (123.77,159.73) .. controls (174.74,121.19) and (353.2,100.21) .. (423.71,155.17) ;
			\draw [shift={(424.76,156)}, rotate = 218.76] [fill={rgb, 255:red, 0; green, 0; blue, 0 }  ][line width=0.08]  [draw opacity=0] (12,-3) -- (0,0) -- (12,3) -- cycle    ;
			\draw   (163.76,159) .. controls (214.74,120.47) and (353.6,100.2) .. (423.71,155.17) ;
			\draw [shift={(424.76,156)}, rotate = 218.76] [fill={rgb, 255:red, 0; green, 0; blue, 0 }  ][line width=0.08]  [draw opacity=0] (12,-3) -- (0,0) -- (12,3) -- cycle    ;
			\draw    (203.76,158.28) .. controls (254.73,119.74) and (354,100.2) .. (423.72,155.17) ;
			\draw [shift={(424.76,156)}, rotate = 218.76] [fill={rgb, 255:red, 0; green, 0; blue, 0 }  ][line width=0.08]  [draw opacity=0] (12,-3) -- (0,0) -- (12,3) -- cycle    ;
			\draw    (285.24,183.72) .. controls (312.86,221.81) and (377.35,237.84) .. (424.53,182.83) ;
			\draw [shift={(425.24,182)}, rotate = 130.15] [fill={rgb, 255:red, 0; green, 0; blue, 0 }  ][line width=0.08]  [draw opacity=0] (12,-3) -- (0,0) -- (12,3) -- cycle    ;
			\draw    (325.24,183) .. controls (352.85,221.08) and (377.75,237.83) .. (424.53,182.83) ;
			\draw [shift={(425.24,182)}, rotate = 130.15] [fill={rgb, 255:red, 0; green, 0; blue, 0 }  ][line width=0.08]  [draw opacity=0] (12,-3) -- (0,0) -- (12,3) -- cycle    ;
			\draw    (365.23,182.27) .. controls (374.02,209.2) and (373.07,236.07) .. (424.46,182.81) ;
			\draw [shift={(425.24,182)}, rotate = 133.83] [fill={rgb, 255:red, 0; green, 0; blue, 0 }  ][line width=0.08]  [draw opacity=0] (12,-3) -- (0,0) -- (12,3) -- cycle    ;
			\draw    (438,168.76) -- (489,168.52) ;
			\draw [shift={(491,168.51)}, rotate = 179.73] [fill={rgb, 255:red, 0; green, 0; blue, 0 }  ][line width=0.08]  [draw opacity=0] (12,-3) -- (0,0) -- (12,3) -- cycle    ;

			\draw    (337.97,265.01) .. controls (360,260) and (400.02,240) .. (425.76,182) ;
			\draw [shift={(425.76,182)}, rotate = 114] [fill={rgb, 255:red, 0; green, 0; blue, 0 }  ][line width=0.08]  [draw opacity=0] (12,-3) -- (0,0) -- (12,3) -- cycle    ;
			
			\draw    (337.97,84.13) .. controls (360,90) and (400.02,120) .. (425.76,156) ;
			\draw [shift={(425.76,156)}, rotate = 235] [fill={rgb, 255:red, 0; green, 0; blue, 0 }  ][line width=0.08]  [draw opacity=0] (12,-3) -- (0,0) -- (12,3) -- cycle    ;

			\draw    (503.76,154) .. controls (503,41.57) and (244.6,-53.04) .. (244.98,70.13) ;
			\draw [shift={(245,72)}, rotate = 268.63] [fill={rgb, 255:red, 0; green, 0; blue, 0 }  ][line width=0.08]  [draw opacity=0] (12,-3) -- (0,0) -- (12,3) -- cycle    ;
			\draw    (503.76,154) .. controls (503,41.57) and (323.81,-53.04) .. (324.97,70.13) ;
			\draw [shift={(325,72.01)}, rotate = 268.63] [fill={rgb, 255:red, 0; green, 0; blue, 0 }  ][line width=0.08]  [draw opacity=0] (12,-3) -- (0,0) -- (12,3) -- cycle    ;
			\draw    (503.76,155.28) .. controls (503,42.84) and (163.42,-53.04) .. (162.98,70.12) ;
			\draw [shift={(163,71.99)}, rotate = 268.63] [fill={rgb, 255:red, 0; green, 0; blue, 0 }  ][line width=0.08]  [draw opacity=0] (12,-3) -- (0,0) -- (12,3) -- cycle    ;
			\draw    (504.23,181.27) .. controls (505.99,298.27) and (164,452) .. (163,274.99) ;
			\draw [shift={(163,274.99)}, rotate = 89.68] [fill={rgb, 255:red, 0; green, 0; blue, 0 }  ][line width=0.08]  [draw opacity=0] (12,-3) -- (0,0) -- (12,3) -- cycle    ;
			\draw    (504.23,181.27) .. controls (505.99,298.27) and (246,453.01) .. (245,276) ;
			\draw [shift={(245,276)}, rotate = 89.68] [fill={rgb, 255:red, 0; green, 0; blue, 0 }  ][line width=0.08]  [draw opacity=0] (12,-3) -- (0,0) -- (12,3) -- cycle    ;
			\draw    (504.23,181.27) .. controls (505.99,298.27) and (326,453.01) .. (325,276.01) ;
			\draw [shift={(325,276.01)}, rotate = 89.68] [fill={rgb, 255:red, 0; green, 0; blue, 0 }  ][line width=0.08]  [draw opacity=0] (12,-3) -- (0,0) -- (12,3) -- cycle    ;

			\draw (498,162) node [anchor=north west][inner sep=0.75pt]  [xscale=0.87,yscale=0.87] [align=left] {\footnotesize{a}};
			\draw (420,162) node [anchor=north west][inner sep=0.75pt]  [xscale=0.87,yscale=0.87] [align=left] {\scriptsize{p}};
			\draw (58,33) node [anchor=north west][inner sep=0.75pt]  [xscale=0.87,yscale=0.87]  {\scriptsize{$c_{1}^{1}$}};
			\draw (58,73) node [anchor=north west][inner sep=0.75pt]  [xscale=0.87,yscale=0.87]  {\scriptsize{$c_{1}^{2}$}};
			\draw (58,113) node [anchor=north west][inner sep=0.75pt]  [xscale=0.87,yscale=0.87]  {\scriptsize{$c_{1}^{3}$}};
			\draw (57,210) node [anchor=north west][inner sep=0.75pt]  [xscale=0.87,yscale=0.87]  {\scriptsize{$c_{2}^{1}$}};
			\draw (57,249) node [anchor=north west][inner sep=0.75pt]  [xscale=0.87,yscale=0.87]  {\scriptsize{$c_{2}^{2}$}};
			\draw (57,290) node [anchor=north west][inner sep=0.75pt]  [xscale=0.87,yscale=0.87]  {\scriptsize{$c_{2}^{3}$}};
			\draw (154,73) node [anchor=north west][inner sep=0.75pt]  [xscale=0.87,yscale=0.87]  {\scriptsize{$x_{1}^{1}$}};
			\draw (236,73) node [anchor=north west][inner sep=0.75pt]  [xscale=0.87,yscale=0.87]  {\scriptsize{$x_{1}^{2}$}};
			\draw (316,73) node [anchor=north west][inner sep=0.75pt]  [xscale=0.87,yscale=0.87]  {\scriptsize{$x_{1}^{3}$}};
			\draw (153,250) node [anchor=north west][inner sep=0.75pt]  [xscale=0.87,yscale=0.87]  {\scriptsize{$x_{2}^{1}$}};
			\draw (236,251) node [anchor=north west][inner sep=0.75pt]  [xscale=0.87,yscale=0.87]  {\scriptsize{$x_{2}^{2}$}};
			\draw (316,251) node [anchor=north west][inner sep=0.75pt]  [xscale=0.87,yscale=0.87]  {\scriptsize{$x_{2}^{3}$}};
			\draw (115,161.4) node [anchor=north west][inner sep=0.75pt]  [xscale=0.87,yscale=0.87]  {\scriptsize{$n_{1}^{1}$}};
			\draw (154,160.4) node [anchor=north west][inner sep=0.75pt]  [xscale=0.87,yscale=0.87]  {\scriptsize{$n_{1}^{2}$}};
			\draw (194,160.4) node [anchor=north west][inner sep=0.75pt]  [xscale=0.87,yscale=0.87]  {\scriptsize{$n_{1}^{3}$}};
			\draw (275,158) node [anchor=north west][inner sep=0.75pt]  [xscale=0.87,yscale=0.87]  {\scriptsize{$n_{2}^{1}$}};
			\draw (315,158) node [anchor=north west][inner sep=0.75pt]  [xscale=0.87,yscale=0.87]  {\scriptsize{$n_{2}^{2}$}};
			\draw (355,158) node [anchor=north west][inner sep=0.75pt]  [xscale=0.87,yscale=0.87]  {\scriptsize{$n_{2}^{3}$}};
		\end{tikzpicture}
		\vspace*{-12mm}
		\caption{Weighted majority graph corresponding to the election $(C,V')$ constructed from a \textsc{3SAT} instance in Example~\ref{ex:ccdc}.
		All edges have a weight of two.}
		\label{fig:congra}
	\end{figure}
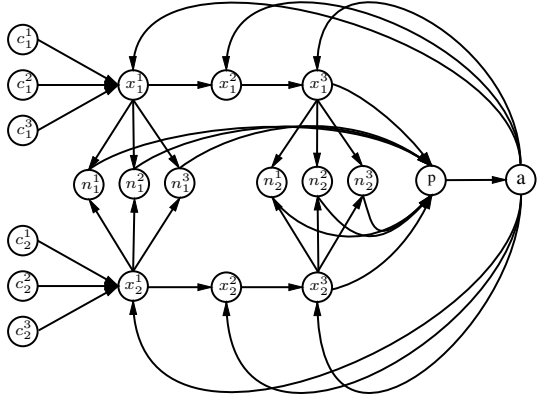

	Before control, every clause candidate $c_i^j$, with  $i \in [2]$ and $j \in [3]$, ties the other clause candidates and has a  path to every other candidate while no candidate has a path to this clause candidate.\footnote{Note that all paths have a strength of two, as there are no other edge weights in the weighted majority graph.}
	Each remaining candidate, including the distinguished candidate~$p$, ties every other candidate except for the six clause candidates, to whom they lose.
	Thus the six clause candidates $c_1^1$, $c_1^2$, $c_1^3$, $c_2^1$, $c_2^2$, and $c_2^3$ are the Schulze winners of the election.
	Now we want to delete at most two candidates to make $p$ a winner of the election.
	Consider the assignment $x_1=\textsc{true}$, which satisfies the first clause, and $x_2=\textsc{true}$, which satisfies the second clause of the CNF formula $\varphi$, i.e., we have a yes-instance of \textsc{3SAT}.\footnote{The truth assignment of $x_3$ is irrelevant, as both clauses and thus the formula $\varphi$ are already true due to the truth assignment of $x_1$ and~$x_2$.} 
	To ensure that $p$ is a winner of the election, it is necessary to destroy each path from a clause candidate to $p$.
	Menton and Singh~\shortcite{men-sin:c:control-complexity-schulze} argue that in order to do so, it is sufficient to ``{\normalfont delete one literal candidate for each clause, selecting a literal that is satisfied by the satisfying assignment for $Cl$}.''
	However, if we delete the literal candidates $x_1^1$ and $x_2^2$, there is still a  path from the candidates $c_2^1,c_2^2,$ and $c_2^3$, traversing $x_2^1$ and $n_1^i$ with $1\leq i\leq 3$, to the distinguished candidate $p$.
	Therefore, $\strongestpath{c_2^i}{p} > \strongestpath{p}{c_2^i}$ for $i \in [3]$ and thus $p$ is still not a Schulze winner of the election.
	
	We now argue that $p$ cannot be made a Schulze winner of the election by deleting any other possible choice of $k = 2$ candidates, so we indeed have a no-instance of Schulze-\textsc{CCDC}.
	First, there are $k+1 = 3$ candidates in each group of clause candidates $C_i$ and negation candidates $N_{ijmn}$, and since all of them have the same incoming and outgoing edges in the WMG, deleting only two of them cannot make $p$ a Schulze winner of the election.
	Further, deleting candidate $a$ would result in $p$ losing to all literal candidates in~$L$.
	Therefore, the only way to guarantee $p$`s victory is to delete two of the literal candidates $x_1^{j}, x_2^{j}$ for $j \in [k+1]$. 
	It is easy to see that at least one literal candidate for each clause must be deleted to break the paths between all clause candidates and~$p$.
	It thus suffices to consider only those pairs of literal candidates where $i=1$ for one and $i=2$ for the other.
	The following table lists those candidates to whom $p$ loses when the pair in the corresponding column is deleted, thus preventing $p$ from becoming a Schulze winner of the election:
	{
		\noindent
		\begin{tabularx}{\columnwidth}{lCCCCCCCCc}
			\multirow{ 2}{*}{deleted pair} & $x_1^1$ & $x_1^1$ & $x_1^1$ & $x_1^2$ & $x_1^2$ & $x_1^2$ & $x_1^3$ & $x_1^3$ & $x_1^3$ \\
			
			& $x_2^1$ & $x_2^2$ & $x_2^3$ & $x_2^1$ & $x_2^2$ & $x_2^3$ & $x_2^1$ & $x_2^2$ & $x_2^3$ \\
			\midrule
			wins against $p$ & $n^j_1$ & $c^j_2$ & $c^j_2$ & $c^j_1$ & $c^j_i$ & $c^j_i$ & $c^j_1$ & $c^j_i$ & $c^j_i,~n^j_2$
		\end{tabularx}
	}
	\mbox{}\\
	Note that $j \in [k+1]$ and $i \in [2]$.
	This shows that a yes-instance of $\threesat$ 
	has been mapped to a no-instance of Schulze-\textsc{CCDC} by the reduction in the proof of \cite[Thm.~2.2]{men-sin:c:control-complexity-schulze}.
\end{example}

\sproofof{Theorem~\ref{thm:schulze-ccdc-unique}}
	The proof is analogous to that of Theorem~\ref{thm:schulze-ccdc}, only the voter lists must be adapted as follows:
	
	{
		\vspace{1em}
		\noindent
		\begin{tabularx}{\columnwidth}{Cll}
			$\#$ & preferences & for each \\
			\midrule
			$3$ & $\voteW{c_i^j}{x_i^1}$ & $i\in [k]$, $j\in [k+1]$\\
			$3$ & $\voteW{x_i^1}{x_i^2}$ & $i\in [k]$\\
			$3$ & $\voteW{x_i^2}{x_i^3}$ & $i\in [k]$\\
			$3$ & $\voteW{x_i^3}{p}$ & $i \in [k]$\\
			$4$ & $\voteW{a}{x}$ & $x\in L$\\
			$4$ & $\voteW{p}{a}$&\\
			$1$ & $\voteW{n}{p}$ & $n\in N$\\ 
			$2$ & $\voteW{a}{c}$ & $c\in K$\\
			$2$ & $\voteW{x_i^j}{n^l_{ijmn}}$& $l \in [k+1]$ where $x_m^n$ is the negation of $x_i^j$
		\end{tabularx}
	}
\eproofof{Theorem~\ref{thm:schulze-ccdc-unique}}

\section{Variations of Destructive Control in Schulze Elections}

Control with the additional twist of candidate groups was introduced by Erd{\'e}lyi et al.~\shortcite{erd-hem-hem:c:more-natural-models-of-electoral-control-by-partition} as one of their ``more natural'' models of electoral control.
Intuitively, all candidates are assigned some group (i.e., a label) and control actions can only consider whole groups of candidates.
For example, one cannot simply delete a candidate without also deleting all other candidates of that same group.
As noted by Erd{\'e}lyi et al.~\shortcite{erd-hem-hem:c:more-natural-models-of-electoral-control-by-partition}, resistance to the standard control problem carries over to the variant with groups.
However, the reverse is not always true.
A voting rule may very well be vulnerable to the standard control problem, but resistant once groups are introduced.
In the following, we will denote the group variant of a control problem by appending the suffix \textsc{G}, e.g., \textsc{DCACG} for \textsc{Destructive Control by Adding Candidate Groups}.

Ranked pairs voting is resistant to standard variants of manipulation, bribery, and control by deleting as well as adding
candidates.
All subsequently considered control problems are either solved in previous sections (Corollary~\ref{cor:rp-dcrc}) or resistance carries over from the standard control problems.
Consequently, in this section we focus on the Schulze method only.

Following the approach by Menton and Singh~\shortcite{men-sin:c:control-complexity-schulze}, we study the relation of the given problems to variations of minimum vertex cut.
First, we define two variations of \ppvclong, called \mippvclong\ and \cvclong. 

In \mippvclong, we are looking for a set of vertices such that no path from $s$ to $t$ is present while a path from $t$ to $s$ remains, with the additional twist that one must include a given number of labeled vertices.

\EP{\mippvclong}
{A directed graph $G=(V,E)$, where $V$ contains some labeled vertices $V_l \subset V$, two vertices $s, t \in V$ with $s\neq t$, and two integers
	$x$ and~$y$.}
{Is there a subset $V' \subseteq V$ such that $V'$ is a path-preserving vertex cut and $\cardinality{V' \setminus V_l} \le x$ and $\cardinality{V' \cap V_l} \ge y$?}

In \cvclong, we are given a graph with colored vertices and each color must either be fully included in the cut set or not at all.
Note that this also applies to $s$ and~$t$, and
vertices with their colors
must not be in the cut set.
For a vertex~$v$, we define $\mathrm{col}(v)$ to return
$v$'s color,
and we define $V[\mathrm{col}] \subseteq V$ to be the set of vertices with the given color~$\mathrm{col}$.

\EP{\cvclong}
{A directed graph $G=(V,E)$, where each vertex is colored by $\mathrm{col}$, two vertices $s, t \in V$ with $s\neq t$, and an integer~$k$.}
{Is there a subset $V' \subseteq V$ such that $V'$ is a path-preserving vertex cut, $\cardinality{V'} \le k$, and $V[\mathrm{col}(v)] \subseteq V'$ for each $v \in V'$?}

To show that Schulze is \emph{vulnerability-combining} (which according to \cite[Definition 4.8]{fal-hem-hem:j:multimode-control} means that combining only vulnerable prongs leads to vulnerability to the resulting multimode attacks) and thus to fully classify multimode control in Schulze elections, it suffices to show that Schulze-\textsc{DCAC+DC} is polynomial-time solvable.
By reducing Schulze-\textsc{DCAC+DC} to \mippvclong, we transfer this challenge to solving whether \mippvclong\ is solvable in polynomial time.

\begin{proposition}\label{prop:ccrc-min-cut}
	In the nonunique-winner model, \schulzeproblem{DCAC+DC} polynomial-time Turing-re\-duces to \mippvclong.
\end{proposition}
\begin{proof}
	Consider a \schulzeproblem{DCAC+DC} instance
	$(C, D, V, p, \lAC, \lDC)$.
	We start by checking the trivial cases.
	First, if $p$ is not a Schulze winner
	of $(C,V)$, then
	the destructive goal of the chair is accomplished without having to add or delete any candidate.
	Second, if $p$ is a Condorcet winner of the election $(C \cup D, V)$, then control is
	impossible.\footnote{Note that Condorcet voting is immune to destructive control by deleting candidates~\cite{hem-hem-rot:j:destructive-control} (i.e., no matter which candidates are deleted from~$C \setminus \{p\}$, a Condorcet winner $p$ cannot be dethroned).  Further, no matter which candidates from $D$ are added to~$C$ (even if all candidates from $D$ are added), $p$ remains a Condorcet winner.  Finally, Schulze is Condorcet-consistent, so a Condorcet winner is also a Schulze winner.}
	If none of these cases are true, we continue as follows.
	We first label all candidates from $D$ and then add them to the given election $(C,V)$ to receive a weighted majority graph, where all candidates from $D$ are labeled, while all original candidates from $C$ remain unlabeled.
	Now we iterate through all candidates to check if we can make $p$ lose against one of them.
	This can now be done similarly to the proof of
	\cite[Lemma~2.4]{men-sin:c:control-complexity-schulze}, except for the following: Instead of applying a subroutine for
	solving \ppvclong, we apply a subroutine for solving \mippvclong, where
	$x = \lDC$ and $y = \vert D \vert - \lAC$.
	We return ``yes'' if we ever receive a positive result from this call;
	otherwise, we return ``no.''
	This is correct, as
	we return ``yes'' if and only if we have deleted (at most) $\lDC$ unlabeled candidates while (at most) $\lAC$ candidates from $D$ remain in the controlled election.
\end{proof}

In \textsc{DCAC+DC}, the limit of how many candidates may be added, $\lAC$, and the limit of how many candidate may be deleted, $\lDC$, bear no relation to one another.
For the exact variant of this multimode control problem and for control by replacing candidates, this is different, though, and we need to reflect this in the vertex cut.
For the reduction to apply to \schulzeproblem{EDCAC+DC} and \schulzeproblem{DCRC}, we thus need to change how the limits in the vertex cut are handled.
That is, we need to (a)~set an exact number of unlabeled as well as an exact number of labeled vertices included in the vertex cut for \schulzeproblem{EDCAC+DC}, and (b)~set these two numbers to be actually the same for \schulzeproblem{DCRC}.

\begin{proposition}\label{prop:ccdvg-min-cut}
	In the nonunique-winner model, \schulzeproblem{DCDCG} and \schulzeproblem{DCACG} poly\-no\-mi\-al-time Turing-reduces to \cvclong. 
\end{proposition}
\begin{proof}
	Our proof for the reduction of \dcdcg\ and \dcacg\ works similar to the proof of
	\cite[Lemma~2.4]{men-sin:c:control-complexity-schulze}, but instead of applying a subroutine for
	solving \ppvclong, we apply one for solving \cvclong.
\end{proof}

While we do not know the complexity of \textsc{Maximum-Inclusion Path-Preserving Vertex Cut} and \textsc{Colored Path-Preserving Vertex Cut} yet, by Propositions~\ref{prop:ccrc-min-cut} and~\ref{prop:ccdvg-min-cut}
membership of these variants of vertex cut in~\p\
immediately gives \p\ results for the control problems \textsc{DCDCG} and
\textsc{DCACG} in Schulze elections.
Note that
the standard vertex cut problem can be solved in polynomial time~\cite{wes:b:introduction-graph-theory}.

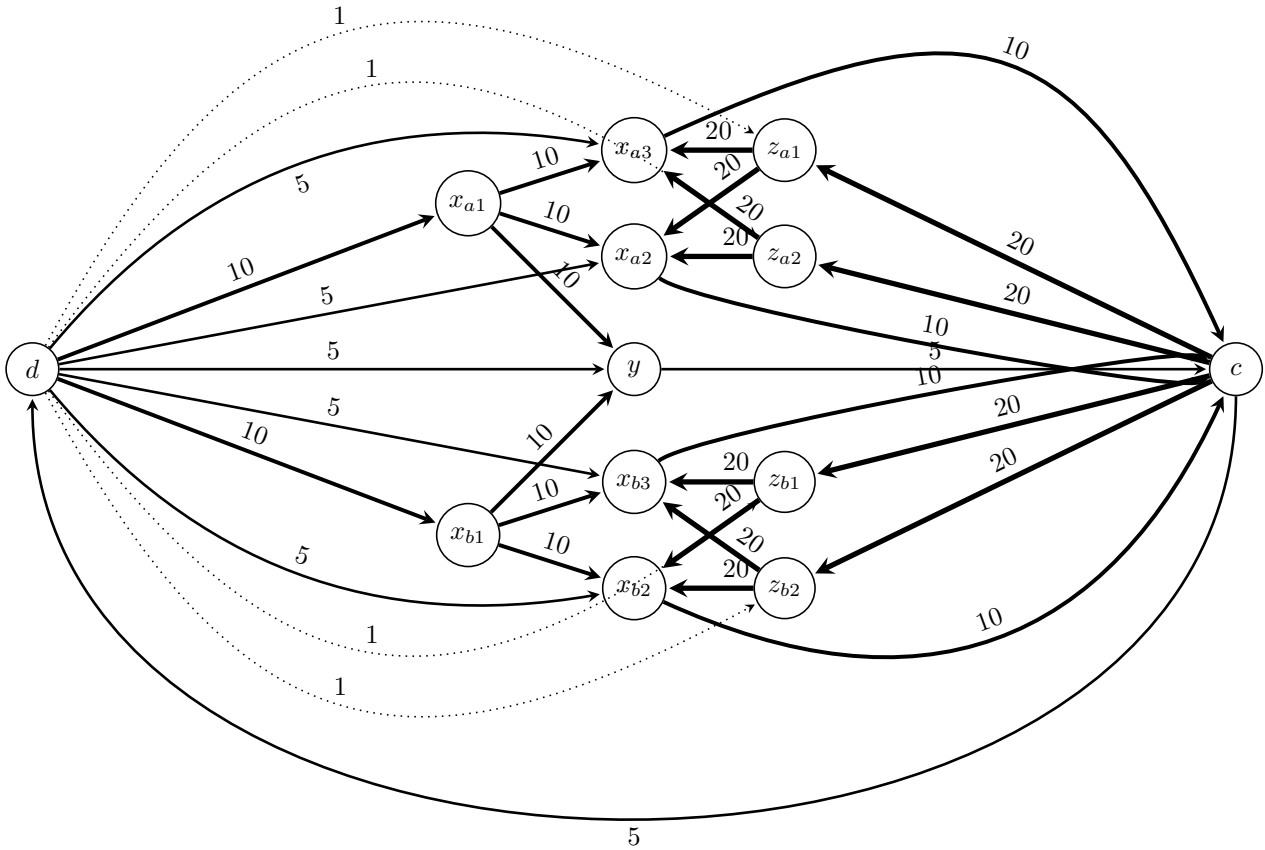
\begin{figure*}
	\begin{tikzpicture}
		[node distance=1cm, semithick]			
		\tikzstyle{voteedge}=[->,>=stealth,shorten >=1pt]
		\tikzstyle{candidate}=[minimum size=0.7cm, draw, circle]

		\node[candidate](xa3)										{$x_{a3}$};
		\node[candidate](xa1)	[below left of=xa3, xshift=-1.5 cm]	{$x_{a1}$};
		\node[candidate](xa2)	[below right of=xa1, xshift=1.5 cm]	{$x_{a2}$};
		\node[candidate](y) 	[below of =xa2, yshift=-0.5cm]						{$y$};
		
		\node[candidate](xb3)	[below of =y,yshift=-0.5cm]						{$x_{b3}$};
		\node[candidate](xb1)	[below left of=xb3, xshift=-1.5 cm]	{$x_{b1}$};
		\node[candidate](xb2)	[below right of=xb1, xshift=1.5 cm]	{$x_{b2}$};
		
		\node[candidate](za1)	[right of =xa3, xshift=1cm]			{$z_{a1}$};
		\node[candidate](za2)	[right of =xa2, xshift=1cm]			{$z_{a2}$};
		\node[candidate](zb1)	[right of =xb3, xshift=1cm]			{$z_{b1}$};
		\node[candidate](zb2)	[right of =xb2, xshift=1cm]			{$z_{b2}$};
		
		\node[candidate](d) 	[left of =y, xshift=-7cm]			{$d$};
		\node[candidate](c) 	[right of =y, xshift=7cm, ]	{$c$};

		\foreach \x/\bend/\lpos in {xa2//above,xa3/bend left/below,xb2/bend right/above,xb3//above}
		\draw[voteedge,line width=1pt] (d) to [\bend] node[midway, sloped, \lpos] {$5$}  (\x);
		
		\foreach \x in {xa1,xb1}
		\draw[voteedge, line width=1.5pt] (d) -- (\x) node[midway, sloped, above] {$10$};

		\foreach \x in {xa1,xb1}
		\draw[voteedge, line width=1.5pt] (\x) -- (y) node[midway, sloped, above] {$10$};
		
		\foreach \x/\z/\pos in {xa2/za2/0.2,xa3/za1/0.4,xa3/za2/0.2,xa2/za1/0.2,
			xb2/zb2/0.2,xb3/zb1/0.2,xb3/zb2/0.2,xb2/zb1/0.2}
		\draw[voteedge, line width=2pt] (\z) -- (\x) node[pos=\pos, sloped, above] {$20$};

		\foreach \z in {za1,za2,zb1,zb2}
		\draw[voteedge, line width=2pt] (c) -- (\z) node[midway, sloped, above] {$20$};

		\foreach \x/\bend/\ln in {xa2/bend right/0.2,xb3/bend left/0.2}
		\draw[voteedge,looseness=\ln,line width=1.5pt] (\x) to [\bend] node[midway, sloped, above] {$10$} (c);
		\draw[voteedge,looseness=1.5,line width=1.5pt] (xa3) to [bend left = 45] node[midway, sloped, above] {$10$} (c);
		\draw[voteedge,looseness=1.2,line width=1.5pt] (xb2) to [bend right = 45] node[midway, sloped, above] {$10$} (c);

		\foreach \x/\y in {xa1/xa2,xa1/xa3,xb1/xb2,xb1/xb3}
		\draw[voteedge, line width=1.5pt] (\x) -- (\y) node[midway, sloped, above] {$10$};
		
		\draw[voteedge, line width=1pt] (d) to node[midway, sloped, above] {$5$} (y);
		\draw[voteedge, line width=1pt] (y) to node[midway, sloped, above] {$5$} (c);
		\draw[voteedge, looseness=1.2, line width=1pt] (c) to [bend left=90] node[midway, below] {$5$} (d);

		\foreach \x/\bend in {za1/bend left,za2/bend left,zb1/bend right,zb2/bend right}
		\draw[voteedge,dotted, looseness = 1.4] (d) to [\bend= 45] node[midway, above] {$1$}  (\x);
		
	\end{tikzpicture}
	\caption{A WMG of an election where for $\ell = 2$ control by deleting candidates is possible in the unique-winner model but impossible in the nonunique-winner model. The edge width visually illustrates the edge weight.}
	\label{fig:election-dcdc-unique-vs-nonunique}
\end{figure*}

\section{Schulze Is Vulnerable to Destructive Control by Deleting Candidates}\label{apdx:schulzeDCDC}

In this section, we give a polynomial-time algorithm for solving Schulze-\textsc{DCDC} in the nonunique-winner model.

Consider an election $(C,V)$ and a corresponding control instance $((C,V),d,\ell)$ of Schulze-\textsc{DCDC}.
If the despised candidate $d$ initially is a Schulze winner of $(C,V)$, our goal is to find a candidate $c$ with a stronger path to $d$ than $d$ has to $c$ by deleting at most $\ell$ candidates.
The algorithm works as follows.

First, if the despised candidate $d$ is already not a Schulze winner of $(C,V)$, return true.
Otherwise, iterate over the set of candidates $C\setminus \{d\}$.
For each candidate $c\in C\setminus\{d\}$, we check whether $c$ is a possible rival of~$d$.
If $d$ beats $c$ directly, i.e., the edge from $d$ to $c$ is stronger than any path from $c$ to~$d$, we exclude $c$ as a rival and move on to the next possible candidate.
Additionally, if $P[c,d] = 0$, we also exclude~$c$ and move on.
Otherwise, $c$ is a possible rival of $d$ and we move on to the deletion stage and initialize a deletion counter $ctr = 0$ for this candidate~$c$. 
Next, for the graph $G$ of all stronger paths from $d$ to~$c$, i.e., all paths where the strength of the path is greater than $P[c,d]$, we check whether deleting the in-neighbors $N^+_G(c)$ of $c$ is possible within our deletion limit and accomplishes our goal of dethrowning~$d$:
We increment the deletion counter by $ctr  += \cardinality{N^+_G(c)}$ and check whether $ctr > \ell$; if so, we move on to the next possible candidate; otherwise, we delete $N^+_G(c)$ from the original election.
If $d$ is not a Schulze winner of the election after deletion, return true.
If $d$ still wins, we have cut a path from $c$ to $d$ by deleting $N^+_G(c)$.
Repeat the above steps until either the deletion limit is reached or $d$ is not a winner of the election anymore.
Finally, if there is no candidate such that deleting at most $\ell$ candidates makes this candidate win against~$d$, return false.

\vspace{1em}
\sproofof{Theorem~\ref{thm:schulze-dcdc-our-result}}
We first show that the algorithm runs in polynomial time, and prove correctness thereafter.
The algorithm contains several subroutines all of which can be executed in polynomial time.
First, the winner problem for Schulze is known to be solvable in polynomial time \cite{sch:j:schulze-voting}.\footnote{\label{footnote:fine-grained complexitySchulze}The Schulze voting method was recently examined through the lens of fine-grained complexity by Sornat et al.~\shortcite{sor-vas-xu:c:fine-grained-complexity-schulze}.}
Then we use a subroutine to construct a graph of the stronger paths from the despised candidate to some other candidate $c \in C$, a possible rival of~$d$.
This construction can be done in polynomial time by adapting the algorithm for solving the winner problem for Schulze, which itself is an adaption of the Floyd--Warshall algorithm for determining shortest paths in a directed weighted graph.
Finally, the loop over all candidates is bounded by the number of candidates, and the execution of the subroutines in the loop is bounded by the deletion limit~$\ell$.
Overall, the algorithm runs in polynomial time.

We now show that the algorithm will always return true if $((C,V),d,\ell)$ is a yes-instance of Schulze-\textsc{DCDC}, and will return false otherwise.

Assume that $((C,V),d,\ell)$ is a yes-instance of Schulze-\textsc{DCDC}, that is, there exists a subset of candidates $C' \subset C$ such that $\cardinality{C'} \le \ell$ and $d$ is not a Schulze winner of the election $(C\setminus C', V)$.
By Theorem~\ref{thm:ccdc-neighbors}, we know that it suffices to consider the in-neighborhood of candidates $c \in C \setminus \{d\}$ to determine whether one of them beats $d$ by deleting at most $\ell$ candidates.
Obviously, our algorithm always finds such a candidate (whenever there exists one) and thus returns true.

Assume that $((C,V),d,k)$ is a no-instance of Schulze-\textsc{DCDC}, that is, there exists no subset of candidates $C' \subset C$ with $\cardinality{C'} \le \ell$ such that $d$ is not a Schulze winner of the election $(C\setminus C', V)$.
Since any set of candidates $C'$ for which control is successful must be greater than $\ell$ and by Theorem~\ref{thm:ccdc-neighbors} the smallest such set can be found in the neighborhood of some $c \in C \setminus \{d\}$, we have that the algorithm will always exceed the deletion counter and thus return false after examining every $c \in C \setminus \{d\}$.
\eproofof{Theorem~\ref{thm:schulze-dcdc-our-result}}

Unfortunately, we cannot easily transfer our algorithm to work in the unique-winner model.
In this model, the goal is to prevent the sole victory of a despised candidate~$d$.
In addition to outright beating~$d$, this goal can also be achieved by making a candidate win alongside~$d$.
This leads to a situation where some candidate $c$ ties the despised but the victory of $c$ is dependent on all other candidates in the election.
It is possible that the same election, along with $d$ and the deletion limit~$\ell$, gives a yes-instance in the unique-winner model and a no-instance in the nonunique-winner model.
In Figure~\ref{fig:election-dcdc-unique-vs-nonunique}, we give an example for this scenario.
Candidate $d$ is the unique winner of the election.
It is impossible to prevent $d$ from winning altogether by deleting at most two candidates.
However, by deleting both $x_{a1}$ and $x_{b1}$ we can prevent $d$ from being the sole winner of the election, as candidate $c$ will win alongside~$d$.
Note that this can only be achieved by deleting $x_{a1}$ and $x_{b1}$ neither of which are in-neighbors of~$c$.

\section{Deferred Proofs from Section~\ref{sec:controlbymultimodeandreplacing}}

\sproofof{Lemma~\ref{lem:schulze-ibc}}
	Let $(C,V)$ be an election and $x$ a new candidate.
	Let $(C \cup \{x\},V^x)$ be the election where candidate $x$ is added to every vote in $V$ as the least preferred option.
	Clearly, we have $\pairwise[V^x]{x}{c} = 0$ and thus $\pairwiseDiff[V^x]{x}{c} = -\cardinality{V^x}$ for all $c \in C \setminus \{x\}$. 
	
	For Schulze, it follows that in the WMG candidate $x$ has an incoming edge with weight $\cardinality{V^x}$ for every $c \in C$ and the outdegree of $x$ is~$0$.
	Thus candidate $x$ cannot win and can also not be part of any path between any other candidates $c,c' \in C$. Ultimately, $x$ has no influence whatsoever on the election.
	
	For ranked pairs, it follows that the pair $(c,x)$ will be in the top ranking for each $c \in C$ (possibly among others).
	These pairs will be added in the first round of rankings, leaving a graph where the	vertex $x$ will have no outgoing edges and any edge to $x$ cannot be part of a cycle.
	Therefore, all other rankings and edges in the graph are untouched.
	The winner of election $(C \cup \{x\},V^x)$ is the same as of election $(C,V)$.
\eproofof{Lemma~\ref{lem:schulze-ibc}}

\sproofof{Corollary~\ref{cor:schulze-ccrc}}
	Schulze and ranked pairs are IBC (see Lemma~\ref{lem:schulze-ibc}), from Theorem~\ref{thm:schulze-ccdc} we know that \schulzeproblem{CCDC} is \np-hard, and Parkes and Xia~\shortcite{par-xia:c:strategic-schulze-ranked-pairs} showed \np-hardness of \rpproblem{CCDC}.
	Hence, by the result of Lorregia et al.~\shortcite{lor-nar-ros-ven-wal:c:replacing-candidates}, \schulzeproblem{CCRC} and \rpproblem{CCRC} are also \np-hard.
	It is easy to see that \schulzeproblem{CCRC} and \rpproblem{CCRC} are in \np\ and thus, \np-complete.\footnote{As noted earlier, we use a fixed tie-breaking scheme to ensure tractability.}
\eproofof{Corollary~\ref{cor:schulze-ccrc}}

\sproofof{Corollary~\ref{cor:schulze-eccrc}}
	By Lemma~\ref{lem:ibc-plus-ccdc}, we can extend the proof of Corollary~\ref{cor:schulze-ccrc} to \schulzeproblem{ECCAC+DC}, \schulzeproblem{ECCRC}, \rpproblem{ECCAC+DC}, and \rpproblem{ECCRC}.  
\eproofof{Corollary~\ref{cor:schulze-eccrc}}

\sproofof{Corollary~\ref{cor:rp-dcrc}}
	By Lemma~\ref{lem:schulze-ibc}, ranked pairs is IBC and Parkes and Xia~\shortcite{par-xia:c:strategic-schulze-ranked-pairs} showed \np-hardness for \rpproblem{DCDC}.
	By the result of Lorregia et al.~\shortcite{lor-nar-ros-ven-wal:c:replacing-candidates}, \rpproblem{DCRC} is also \np-hard.
	Since \rpproblem{DCRC} is in \np, it is \np-complete.
	By Lemma~\ref{lem:ibc-plus-dcdc}, this also applies to \rpproblem{EDCAC+DC} and \rpproblem{ECCRC}.  
\eproofof{Corollary~\ref{cor:rp-dcrc}}

\sproofof{Theorem~\ref{theo:votercontrol}}
	It is easy to see that \schulzeproblem{ECCAV+DV}, \schulzeproblem{CCRV}, \schulzeproblem{EDCAV+DV} and \schulzeproblem{DCRV} are in \np.
	To show \np-hardness, we reduce from the $\np$-complete problem
	\textsc{Restricted Exact Cover By 3-Sets} (\rxthreec)~\cite{gon:j:exact-cover-constrained}:
	Given a set $B = \{b_1, \ldots, b_{3s}\}$ with $s \geq 1$ and a list $\mathcal{S} = \{S_1, \ldots, S_{3s}\}$, where $S_i = \{b_{i,1}, b_{i,2}, b_{i,3}\}$ and $S_i \subseteq B$ for all $S_i \in \mathcal{S}$ and each $b_j$ is contained in exactly three sets $S_i \in \mathcal{S}$, does there exist an exact cover, i.e., a sublist $\mathcal{S}' \subseteq \mathcal{S}$ such that each $b_i \in B$ occurs in exactly one $S_i \in \mathcal{S}'$?
	
	Let $(B,\mathcal{S})$ be an \rxthreec\ instance, such that $3s = n$.
	We first present the reduction in the nonunique-winner model;
	the slightly adapted reduction for the unique-winner model will be considered later.
	Let $\lAV = \lDV = s$ for the \schulzeproblem{ECCAV+DV} and \schulzeproblem{EDCAV+DV} instances we construct, and let $\lRV = s$ for the \schulzeproblem{CCRV} and \schulzeproblem{DCRV} instances.
	Further, let $\lConst \gg s$ be a constant much greater than~$s$.\footnote{The precise value of $\lConst$ is not important; all that matters is that when used as an edge weight in a WMG, $\lConst$ is large enough, such that any edge changed by the control actions must not change in direction, i.e., the sign of the edge weight must not flip.
		Recall that the strength of a path in a WMG is specified as the weight of the weakest edge on the path.}
	From $(B,\mathcal{S})$ we construct an election $(C,V)$ as follows.
	Let the candidate set be
	\[
	C = B \cup \{p,w\}.
	\]
	The list of votes contains $\lDV$ votes of the form $w\,B\,p$ and the remaining votes are constructed such that	
	\begin{itemize}
		\item $p$ beats $w$ by $2\lConst$ votes,
		\item each $b_i$ beats $p$ by $2\lConst+2\lDV+2\lAV-2$ votes,
		\item $w$ beats each $b_i$ by votes $\gg2\lConst$ (so as to make these edges not relevant for the strength of a path), and 
		\item all other pairwise differences are $0$ or at least smaller than~$\lConst$.
	\end{itemize}
	
	The WMG of the resulting election is shown in Figure~\ref{fig:apdx:ccrv-election-construction}.
	Note that candidate $w$ is the unique winner of the election.
	
	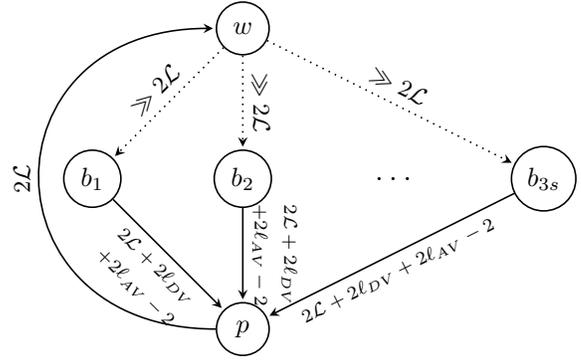
\begin{figure}
		\centering
		\begin{tikzpicture}[node distance=2cm, semithick]
			
			\tikzstyle{voteedge}=[->,>=stealth,shorten >=1pt]
			\tikzstyle{gt2ledge}=[->,>=stealth,shorten >=1pt, dotted]
			\tikzstyle{candidate}=[minimum size=0.7cm, draw, circle]
			\tikzstyle{dots}=[draw=none]
			
			\node[candidate](b2) 							{$b_2$};
			\node[candidate](b1) 		[left of=b2] 		{$b_1$};
			\node[dots](bdots) 			[right of=b2] 		{\Large$\ldots$};
			\node[candidate](b3s) 		[right of=bdots] 	{$b_{3s}$};
			\node[candidate](p) 		[below of=b2] 		{$p$};

			\node[candidate, node distance=2cm](w)         [above of=b2] {$w$};
			
			\draw[gt2ledge] (w) -- node[midway, above, sloped] {\footnotesize{$\gg 2\lConst$}} (b1);
			\draw[gt2ledge] (w) -- node[midway, above, sloped] {\footnotesize{$\gg 2\lConst$}} (b2);
			\draw[gt2ledge] (w) -- node[midway, above, sloped] {\footnotesize{$\gg 2\lConst$}} (b3s);
			
			\draw[voteedge,looseness=2] (p) to [bend left=90] node[midway, above, sloped] {\footnotesize{$2\lConst$}} (w);

			\draw[voteedge] (b1) -- node[align=left, sloped, midway, below] { {\scriptsize $2\lConst+2\lDV$}\\ {\scriptsize $+2\lAV-2$}} (p);
			\draw[voteedge] (b2) -- node[align=left, sloped, midway,above] { {\scriptsize $2\lConst+2\lDV$}\\ {\scriptsize $+2\lAV-2$}}(p);
			\draw[voteedge] (b3s) -- node[sloped, midway, below] { {\scriptsize $2\lConst+2\lDV+2\lAV-2$}} (p);
			
		\end{tikzpicture}
		\caption{The WMG of the election $(C,V)$ from the proof of Theorem~\ref{theo:votercontrol}.}
		\label{fig:apdx:ccrv-election-construction}
	\end{figure}
	
	The list of additional votes $U$ contains one vote
	\[
	S_i\, p\, (B \setminus S_i)\, w
	\hspace{0.5cm} \text{for each}\ S_i \in \mathcal{S}.
	\]
	Let $p$ be the distinguished candidate for the constructive case and $w$ be the despised candidate for the destructive case.
	We claim that $p$ can be made a Schulze winner (and $w$ can be prevented from being the unique Schulze winner) by adding exactly $\lAV$ voters and deleting exactly $\lDV$ voters if and only if $(B,\mathcal{S})$ is a yes-instance of \rxthreec.
	
	From left to right, let $(B,\mathcal{S})$ be a yes-instance of \rxthreec\ and let $\mathcal{S}' \subset \mathcal{S}$ be an exact cover of~$B$.
	We have $\cardinality{\mathcal{S}'} = s$.
	Remove $\lDV$ votes from $V$ of the form $w\,B\,p$ and add $s$ votes $S_i\, p\, (\mathcal{S}\setminus S_i)\, w$ from~$U$, where $S_i \in \mathcal{S}'$.
	Let $V'$ be the resulting list of votes.
	Due to the construction, each strongest path from any candidate to another one can only be composed of edges where the edge weight is greater than $2\lConst$, and each edge with a weight that is much greater than $2\lConst$ has no influence on the strength of the strongest path.
	Therefore, it suffices to calculate the weight changes of edges to and from~$p$.
	In the added votes from $U$ we have
	$b_j \succ p$ in exactly one vote and $p \succ b_j$ in the remaining $s-1$ votes for each $b_j \in B$.
	Therefore, $p$ gains $(s-1) - 1$ in each pairwise comparison to $b_j \in B$. 
	By adding the $s$ votes from~$U$,
	\begin{itemize}
		\item $p$ gains $s-2 = \lAV-2$ against each $b_j \in B$, and
		\item $p$ gains $s = \lAV$ against~$w$.
	\end{itemize}
	By deleting the $\lDV$ votes from~$V$,
	\begin{itemize}
		\item $p$ gains $\lDV$ against each $b_j \in B$, and
		\item $p$ gains $\lDV$ against~$w$.
	\end{itemize}

	We have $\pairwiseDiff[V']{p}{b_j} = 2\lConst+\lDV+\lAV$ for each $b_j \in B$ and $\pairwiseDiff[V']{p}{w} =2\lConst +\lDV+\lAV$.
	Thus we have $\strPath[V']{p}{c} = 2\lConst+\lDV+\lAV = \strPath[V']{c}{p}$ for each $c \in C \setminus \{p\}$, so $p$ is a Schulze winner of the election $(C, V')$, which resulted from $(C,V)$ by adding $s=\lAV$ voters and deleting $s=\lDV$ voters.
	Both the constructive and destructive control actions were successful in the nonunique-winner model.
	
		\begin{figure}
				\centering
				\begin{tikzpicture}[node distance=2cm, semithick]			
						\tikzstyle{voteedge}=[->,>=stealth,shorten >=1pt]
						\tikzstyle{gt2ledge}=[->,>=stealth,shorten >=1pt, dotted]
						\tikzstyle{candidate}=[minimum size=0.7cm, draw, circle]
						\tikzstyle{dots}=[draw=none]
						
						\node[candidate](b2) 							{$b_2$};
						\node[candidate](b1) 		[left of=b2] 		{$b_1$};
						\node[dots](bdots) 			[right of=b2] 		{\Large$\ldots$};
						\node[candidate](b3s) 		[right of=bdots] 	{$b_{3s}$};
						\node[candidate](p) 		[below of=b2] 		{$p$};
						
						\node[candidate, node distance=2cm](w)         [above of=b2] {$w$};
						
						\draw[gt2ledge] (w) -- node[midway, above, sloped] {{\footnotesize $\gg 2\lConst$}} (b1);
						\draw[gt2ledge] (w) -- node[midway, above, sloped] {{\footnotesize $\gg 2\lConst$}} (b2);
						\draw[gt2ledge] (w) -- node[midway, above, sloped] {{\footnotesize $\gg 2\lConst$}} (b3s);
						
						\draw[voteedge,looseness=2] (p) to [bend left=90] node[midway, above, sloped] {{\scriptsize $2\lConst+\lDV+\lAV$}} (w);
						
						\draw[voteedge] (b1) -- node[midway, below, sloped] {{\scriptsize $2\lConst+\lDV+\lAV$}} (p);
						\draw[voteedge] (b2) -- node[midway, above,sloped, align=left] {{\scriptsize $2\lConst+\lDV$}\\ {\scriptsize $+\lAV$}} (p);
						\draw[voteedge] (b3s) -- node[midway, below, sloped] {{\scriptsize $2\lConst+\lDV+\lAV$}} (p);
						
					\end{tikzpicture}
				\caption{The WMG of the resulting election $(C,V')$ from the proof of Theorem~\ref{theo:votercontrol} after successful control.}
				\label{fig:ccrv-controlled-election}
			\end{figure}
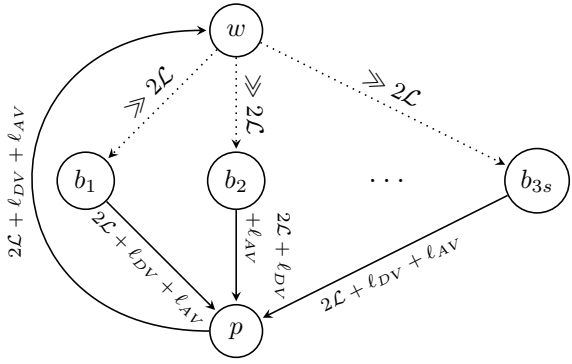
		
	From right to left, let $(B,\mathcal{S})$ be a no-instance of \rxthreec.
	Hence, we know that in any $\mathcal{S}' \subset \mathcal{S}$ with $\cardinality{\mathcal{S}'} = s$ at least one $b_j$ is contained in more than one $S_i \in \mathcal{S}'$.
	The difference between the strength of a strongest path from $w$ to $p$ and the strength of a strongest path from $p$ to $w$ is $\strPath[V]{b_j}{p} - \strPath[V]{p}{b_j} = 2\lDV+2\lAV-2$.
	By deleting $\lDV$ voters, $p$ can make up at most $2\lDV$ and a difference of $2\lAV-2$ remains.
	In any subset $U'$ of $s=\lAV$ votes from $U$ we have one $b_j \in B$ with $\pairwise[U']{b_j}{p} > 1$ and, therefore, $\pairwiseDiff[U']{p}{b_j} \le \lAV-4$.
	It follows that at least one $b_j \in B$ with $\strPath{b_j}{p} \ge 2\lConst+\lDV+\lAV$ exists and, therefore, it holds that $\strongestpath{w}{p} > \strongestpath{p}{w}$, so $p$ cannot be made a Schulze winner by adding $s=\lAV$ votes and deleting $s=\lDV$ votes.
	Since $p$ is the only candidate where the difference in the strength of a strongest path against $w$ is less than $2\lDV+\lAV$, it follows that a sole victory of candidate $w$ cannot be prevented.
	
	For the unique-winner model, we need to slightly adapt the above construction.
	In the voter list~$V$, we need to have every $b_j$ win against $p$ by $2\lConst+2\lDV+2\lAV-3$ votes.
	We then claim that $p$ can be made a \emph{unique} Schulze winner (and $w$ can be prevented from being a winner) by adding exactly $\lAV$ voters and deleting exactly $\lDV$ voters if and only if $(B,\mathcal{S})$ is a yes-instance of \rxthreec.
	The proof of correctness for this slightly modified reduction is analogous.
	Moreover, $\np$-hardness of \schulzeproblem{CCRV} and \schulzeproblem{DCRV} can also be shown analogously in both winner models.
\eproofof{Theorem~\ref{theo:votercontrol}}

\sproofof{Theorem~\ref{theo:rp-votercontrol}}
	We slightly adapt the construction used in the proof of Theorem~\ref{theo:votercontrol} and therefore only provide the overall reduction.
	Instead of having the pairwise differences between each pair of candidates $b_i, b_j \in B$ be zero, we have each $b_i$ win by $\gg 2\lConst$ against all $b_j \in B$, where $i < j$, i.e.,
	\[
	D(b_i,b_j) \gg 2\lConst \hspace{1cm} \forall i < j,\ 1 \le i,j \le 3s.
	\]
	This ensures that no cycles can be formed by inserted edges between any pair of candidates from $B$.
	The rest of the construction remains the same.
	By writing $(c,c')$ we denote the pair of candidates $c$ and $c'$ and---slightly overloading the notation---we use $(c,M)$ to denote all pairs $(c, m)$ where $m \in M$.
		
	In the resulting election, the first ranking of pairwise differences contains all pairs $(w,B)$ and $(b_i, b_j)$, $i < j$, $1 \le i,j \le 3s$.
	As no cycle can be created by adding these pairs, all edges are added.
	The next ranking are all pairs $(B,p)$.
	Again, no cycle can be created and all edges are added.
	However, the remaining pair $(p,w)$ would create an edge from $p$ to~$w$, which would induce a cycle.
	Candidate $w$ remains the sole source and is the unique	ranked pairs winner of the election.
	
	We claim that $p$ can be made a ranked pairs winner (and $w$ can be prevented from being the unique ranked pairs winner) by adding exactly $\lAV$ voters and deleting exactly $\lDV$ voters if and only if $(B,\mathcal{S})$ is a yes-instance of \rxthreec.
	
	From left to right, let $(B,\mathcal{S})$ be a yes-instance of \rxthreec\ and let $\mathcal{S}' \subset \mathcal{S}$ be an exact cover of~$B$.
	Remove $s=\lDV$ votes from $V$ of the form $w\, B\, p$ and add $s = \lAV$ votes $S_i\,  p\, (B\setminus S_i)\, w$ from~$U$, where $S_i \in \mathcal{S}'$.
	The resulting pairwise differences for all pairs are as follows:
	\begin{align*}
		&D(w,B) \gg 2\lConst, \\
		&D(B,p) = 2\lConst + \lDV +\lAV, \text{ and} \\
		&D(p,w) = 2\lConst + \lDV +\lAV.
	\end{align*}
	Assuming a fixed tie-breaking scheme, where the preferred candidate is always favored, we add the pair $(p,w)$  before any of the pairs $(B,p)$, which will then create an edge and thus be skipped.
	Candidate $p$ is the only source and unique ranked pairs winner of the election.
	
	From right to left, let $(B,\mathcal{S})$ be a no-instance of \rxthreec.
	Hence, it holds that in any $\mathcal{S}' \subset \mathcal{S}$ with $\cardinality{\mathcal{S}'} = s$ at least one $b_i$ is contained in more than one $S_i \in \mathcal{S}'$.
	If $D(b_i, p) > D(p,w)$ holds for even one $b_i \in B$, the pair $(b_i, B)$ will be ranked above $(p,w)$ and therefore added.
	When considering $(p,w)$, a path from $w$ to $p$ already exists and therefore the edge is skipped.
	The only way to prevent $w$ from being a unique winner and making $p$ the unique winner is to ensure $D(B, p) \le D(p,w)$.
	By the same argumentation as in the proof of Theorem~\ref{theo:votercontrol}, we cannot make up the pairwise difference unless each $b_i \in B$ is present in exactly one $S_i \in \mathcal{S}'$, which is a contradiction to $(B,\mathcal{S})$ being a no-instance of \rxthreec.
\eproofof{Theorem~\ref{theo:rp-votercontrol}}

\end{document}